\documentclass[sigconf,9pt]{acmart}

\renewcommand\footnotetextcopyrightpermission[1]{} 
\copyrightyear{2024}
\acmYear{2024}
\setcopyright{acmcopyright}\acmConference[ICCAD'24 (to appear)]{Proceedings of the ACM/IEEE International Conference on Computer-Aided Design}{October 27--31, 2024}{New Jersey, USA}
\acmBooktitle{Proceedings of the ACM/IEEE International Conference on Computer-Aided Design (ICCAD), October 27--31, 2024, New Jersey, USA}
\acmPrice{}
\acmDOI{}
\acmISBN{}

\usepackage{graphicx} 

\usepackage{cleveref}
\Crefformat{figure}{#2Figure~#1#3}

\usepackage{xcolor, colortbl}

\usepackage{physics}

\usepackage{tikz}

\usepackage[]{qcircuit}

\newcommand{\libsmall}{\{\mathcal{U}(2), \text{CNOT}\}}

\usepackage{bbm}

\usepackage{amsthm}
\newtheorem{theorem}{Theorem}[section]

\newtheorem{example}[theorem]{Example}

\usepackage{caption}
\usepackage{subcaption}

\usepackage{multicol}
\usepackage{multirow}

\usepackage[titlenumbered,ruled,linesnumbered,noend]{algorithm2e}

\makeatletter
\renewcommand*\env@matrix[1][\arraystretch]{%
  \edef\arraystretch{#1}%
  \hskip -\arraycolsep
  \let\@ifnextchar\new@ifnextchar
  \array{*\c@MaxMatrixCols c}}
\makeatother

\newcommand{\tocheck}[1]{\textcolor{black}{#1}}

\usepackage{soul}

\title{Quantum State Preparation Circuit Optimization Exploiting Don't Cares}

\author{Hanyu Wang, Daniel Bochen Tan, Jason Cong}
\affiliation{%
  \institution{University of California, Los Angeles}
  \city{}
  \country{}
}
\email{{hanyuwang, bctan, cong}@cs.ucla.edu}

\begin{document}
\begin{abstract}
    Quantum state preparation initializes the quantum registers and is essential for running quantum algorithms.
    Designing state preparation circuits that entangle qubits efficiently with fewer two-qubit gates enhances accuracy and alleviates coupling constraints on devices.
    Existing methods synthesize an initial circuit and leverage compilers to reduce the circuit's gate count while preserving the unitary equivalency.
    In this study, we identify numerous conditions within the quantum circuit where breaking local unitary equivalences does not alter the overall outcome of the state preparation (i.e., don't cares).
    We introduce a peephole optimization algorithm that identifies such unitaries for replacement in the original circuit.
    Exploiting these don't care conditions, our algorithm achieves a \tocheck{36\%} reduction in the number of two-qubit gates compared to prior methods.
\end{abstract}

\maketitle

\section{Introduction}
\emph{Quantum states preparation}~(QSP) is indispensable in quantum computing for initializing the state of quantum registers. The initial state to prepare is determined by the specific quantum algorithm and application. Examples of such states include GHZ states, which are instrumental in entanglement experiments~\cite{greenberger1989going}; W states~\cite{dur2000three} and Dicke states~\cite{dicke1954coherence}, which are useful in quantum metrology and act as generators of complex symmetric states~\cite{toth2012multipartite}; VBS states~\cite{murta2023preparing}, valuable for modeling interacting~\cite{read1989valence} quantum spin models. Additionally, certain algorithms encode the problem inputs directly into the initial states~\cite{ashhab2022quantum}. Consequently, developing automated algorithms to design the state preparation circuit is necessary. 

Reducing the gate count, particularly the number of two-qubit gates in the circuit, is essential for improving the performance of noisy intermediate-scale quantum computing. To generate the initial QSP circuit, Boolean methods are proposed utilizing decision diagrams to prepare general $n$-qubit states employing $\mathcal{O}(2^n)$ two-qubit gates~\cite{araujo2021divide, mozafari2019preparation, niemann2016logic}. Besides, specialized algorithms for sparse state preparation are developed to prepare $n$-qubit states with $m$ nonzero amplitudes using $\mathcal{O}(mn)$ two-qubit gates~\cite{gleinig2021efficient, mozafari2022efficient, malvetti2021quantum}. These methods efficiently search for a feasible circuit by decomposing the quantum state using divide-and-conquer techniques. 

To reduce the gate count and optimize the initial circuit, synthesis algorithms partition the circuits into blocks with manageable unitary matrix dimensions~\cite{wu2020qgo} and apply \emph{unitary synthesis} approaches that search for replacement circuits with fewer gates to implement the given unitary~\cite{davis2019heuristics}. These unitary synthesis algorithms can effectively simplify the unitary matrices of systems up to six qubits~\cite{smith2023leap, davis2020towards, younis2020qfast}. Consequently, these algorithms are often employed as peephole optimizations within larger workflows that utilize iterative methods or design space exploration techniques to enhance scalability~\cite{kissinger2019pyzx, staudacher2021optimization}.

Due to the inherent complexity, existing quantum circuit optimization algorithms assume the unitary is fixed for a given circuit, which is not necessarily true in many quantum applications. In the quantum state preparation problem, for instance, all qubits are initialized to the ground state. As a result, only the first column of the unitary matrix corresponding to the ground state affects the circuit outcome, while the entries in all other columns are flexible. Szasz et al. generalize the flexibilities in applications as a multi-set preparation problem and point out various flexible unitaries in Hamiltonian simulation circuits, preparation of general quantum channels, and circuits with ancillary qubits~\cite{szasz2023numerical}. 
Exploiting these flexibilities can significantly improve the circuit performance, as the algorithms are allowed to replace the unitary with a more promising one with a smaller circuit size. However, encoding this flexibility in the optimization algorithm results in an exponentially growing complexity and, thus, does not apply to larger entangled systems~\cite{ashhab2023quantum}. 

In this paper, we propose a novel scalable algorithm to optimize quantum state preparation circuits. Our method formulates the flexibility in modifying unitary operators as don't cares and allows for simpler circuit configurations. Given an initial QSP circuit, we first partition the circuit into \emph{single-target segments} whose functionality can be efficiently expressed. Then, we propagate and derive the don't-care conditions inside these segments. Finally, our resynthesis algorithm identifies and utilizes don't-cares to simplify the segments' implementation. Experimental results show that our approach enhances the optimization of the QSP circuit and reduces the CNOT count by \tocheck{36\%} compared with existing algorithms without employing don't cares. 

In the rest of the paper, we present background and related work in \Cref{sec:background}. Then, we give an example of don't care-based circuit optimization in \Cref{sec:motivating-example} to motivate our work. We illustrate our methodology in \Cref{sec:methodology}. In \Cref{sec:evaluation}, we demonstrate experimental results and evaluate our method.
\begin{figure*}[t]
    \centering
    \begin{subfigure}[b]{.99\linewidth}
    \centering
    \mbox{
    \small
    \Qcircuit @C=1em @R=.3em {
            \mbox{$\mathcal{W}_1$} && \mbox{$\ket{\psi_1}$} &\mbox{$\mathcal{W}_2$}&&& \mbox{$\ket{\psi_2}$} &&&\mbox{$\mathcal{W}_3$}&&&& \mbox{$\ket{\psi_3}$} & \mbox{$\mathcal{W}_4$} & \mbox{$\ket{\psi_4}$} &&\mbox{$\mathcal{W}_5$}&& \mbox{$\ket{\psi_5}$} & \\
        \lstick{q_0} & \gate{\parbox{0.6cm}{\centering \footnotesize $R_y$\\$\pi/2$}} \ar@{--}[]+<1.7em,1em>;[d]+<1.7em,-3em> & \qw & \ctrl{1} & \qw & \ctrl{1} \ar@{--}[]+<.7em,1em>;[d]+<.7em,-3em> & \qw & \ctrl{2} & \qw & \qw & \qw & \ctrl{2} & \qw & \qw \ar@{--}[]+<.7em,1em>;[d]+<.7em,-3em> & \qw \ar@{--}[]+<1.7em,1em>;[d]+<1.7em,-3em> & \qw & \qw & \qw & \qw \ar@{--}[]+<.7em,1em>;[d]+<.7em,-3em> & \qw \\
        \lstick{q_1} & \qw & \gate{\parbox{0.6cm}{\centering \footnotesize $R_y$\\$\pi/4$}} & \targ & \gate{\parbox{0.6cm}{\centering \footnotesize $R_y$\\$-\pi/4$}} & \targ & \qw & \qw & \qw & \ctrl{1} & \qw & \qw & \qw & \ctrl{1} & \gate{\parbox{0.6cm}{\centering \footnotesize $R_y$\\$\pi/2$}} & \qw & \ctrl{1} & \qw & \ctrl{1} & \qw \\
        \lstick{q_2} & \qw & \qw & \qw & \qw & \qw & \gate{\parbox{0.6cm}{\centering \footnotesize $R_y$\\$\pi/4$}} & \targ & \gate{\parbox{0.6cm}{\centering \footnotesize $R_y$\\$-\pi/4$}} & \targ & \gate{\parbox{0.6cm}{\centering \footnotesize $R_y$\\$-\pi/4$}} & \targ & \gate{\parbox{0.6cm}{\centering \footnotesize $R_y$\\$\pi/4$}} & \targ & \qw & \gate{\parbox{0.6cm}{\centering \footnotesize $R_y$\\$\pi/4$}} & \targ & \gate{\parbox{0.6cm}{\centering \footnotesize $R_y$\\$-\pi/4$}} & \targ & \qw
            }
        }
        \caption{Quantum circuit decomposed using conventional quantum compilation algorithms to prepare $\ket{\psi}$.}\label{fig:optimized-conventional}
    \end{subfigure}
    \begin{subfigure}[t]{.6\linewidth}
        \centering
        \mbox{
        \small
            \Qcircuit @C=1em @R=.3em {
                \mbox{$\mathcal{W}_1$} && \mbox{$\ket{\psi_1}$} &
                \mbox{$\mathcal{W}_2$} && \mbox{$\ket{\psi_2}$}
                \mbox{$\mathcal{W}_3$} && \mbox{$\ket{\psi_3}$} \mbox{$\mathcal{W}_4$} \mbox{$\ket{\psi_4}$} &&
                \mbox{$\mathcal{W}_5$} && \mbox{$\ket{\psi_5}$} & \\
        \lstick{q_0} & \gate{\parbox{0.6cm}{\centering \footnotesize $R_y$\\$\pi/2$}} \ar@{--}[]+<1.7em,1em>;[d]+<1.7em,-3em> & \qw & \ctrlo{1} & \qw \ar@{--}[]+<1.7em,1em>;[d]+<1.7em,-3em> & \ctrlo{2} & \qw \ar@{--}[]+<.7em,1em>;[d]+<.7em,-3em> & \qw \ar@{--}[]+<1.7em,1em>;[d]+<1.7em,-3em> & \qw & \qw & \qw & \qw \ar@{--}[]+<.7em,1em>;[d]+<.7em,-3em> & \qw \\
        \lstick{q_1} & \qw & \gate{\parbox{0.6cm}{\centering \footnotesize $R_y$\\$3\pi/4$}} & \targ & \gate{\parbox{0.6cm}{\centering \footnotesize $R_y$\\$-\pi/4$}} & \qw & \ctrlo{1} & \gate{\parbox{0.6cm}{\centering \footnotesize $R_y$\\$\pi/2$}} & \qw & \ctrl{1} & \qw & \ctrl{1} & \qw \\
        \lstick{q_2} & \qw & \qw & \qw & \qw & \targ & \targ & \qw & \gate{\parbox{0.6cm}{\centering \footnotesize $R_y$\\$\pi/4$}} & \targ & \gate{\parbox{0.6cm}{\centering \footnotesize $R_y$\\$-\pi/4$}} & \targ & \qw
            }
        }
        \centering
        \caption{Quantum circuit optimized exploiting don't-cares within the five \\ segments, $\mathcal{W}_1$, ..., $\mathcal{W}_5$ to prepare $\ket{\psi}$.}\label{fig:satisfiability-dc}
    \end{subfigure}
    \begin{subfigure}[t]{.36\linewidth}
        \centering
        \mbox{
        \small
            \Qcircuit @C=.7em @R=.3em {
                && \mbox{$\ket{\psi_1}$} & \mbox{$\ket{\psi'_2}$} &&&&& \mbox{$\ket{\psi_5}$} & \\
        \lstick{q_0} & \gate{\parbox{0.6cm}{\centering \footnotesize $R_y$\\$\pi/2$}} \ar@{--}[]+<1.7em,1em>;[d]+<1.7em,-3em> & \qw \ar@{--}[]+<1.7em,1em>;[d]+<1.7em,-3em> & \qw & \qw & \qw & \ctrl{2} & \qw \ar@{--}[]+<1.7em,1em>;[d]+<1.7em,-3em> & \qw \\
        \lstick{q_1} & \qw & \gate{\parbox{0.6cm}{\centering \footnotesize $R_y$\\$\pi/2$}} & \qw & \ctrl{1} & \qw & \qw & \qw & \qw \\
        \lstick{q_2} & \qw & \qw & \gate{\parbox{0.6cm}{\centering \footnotesize $R_y$\\$\pi/4$}} & \targ & \gate{\parbox{0.6cm}{\centering \footnotesize $R_y$\\$-\pi/2$}} & \targ & \gate{\parbox{0.6cm}{\centering \footnotesize $R_y$\\$\pi/4$}} & \qw
            }
        }
        \caption{Quantum circuit optimized exploiting all don't-cares to prepare $\ket{\psi}$.}\label{fig:observability-dc}
    \end{subfigure}

    \caption{Three quantum circuits with different numbers of gates to prepare the same state $\ket{\psi}$ in \Cref{eqn:target-state}. This example establishes that the gate counts can be reduced by using don't cares in the quantum circuit.}\label{fig:motivating-example}
\end{figure*}
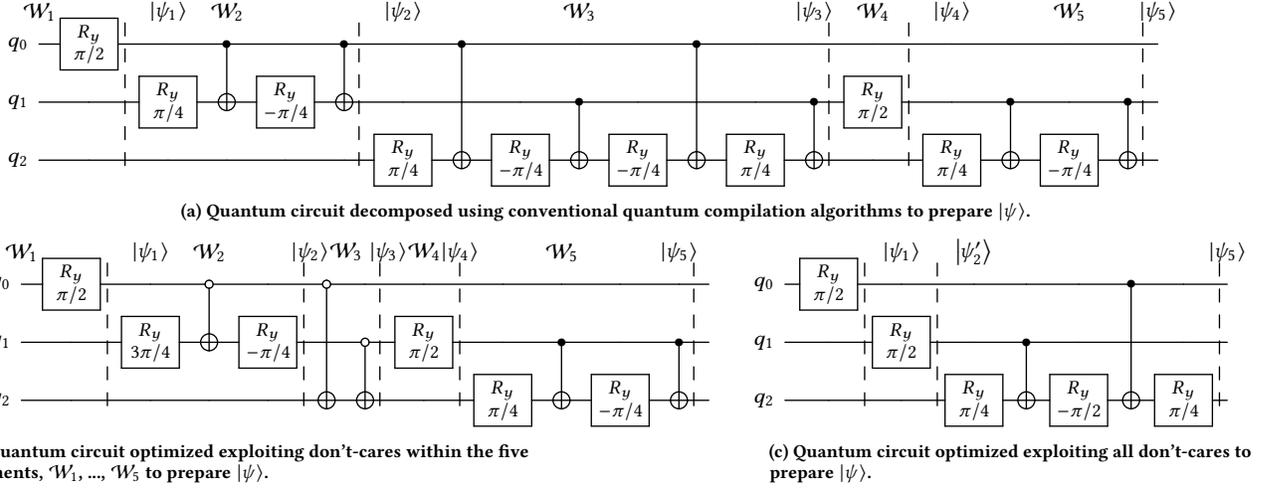

\section{Background}\label{sec:background}
In this section, we provide background for quantum states, gates, and circuits. For clarity and space constraints, we refer readers to established sources for formal definitions of notations~\cite{nielsen2010quantum} and quantum gates~\cite{barenco1995elementary}.

\subsection{Quantum States and Quantum Gates}
We express the $n$-qubit \emph{quantum state} as a linear combination of $2^n$ orthonormal basis vectors.
\begin{equation*}
    \ket{\psi} = \sum_{x\in \{0,1\}^n} c_x \ket{x},\,\text{and}\, \sum_{x\in \{0,1\}^n} |c_x|^2 = 1,
\end{equation*}
where $\ket{x}\in\{0,1\}^n$ are the \emph{basis states}, $c_x\in\mathbb{C}$ are \emph{amplitudes} that indicate the probability of observing $\ket{x}$ after a measurement.

\emph{Quantum gates}, or \emph{operators}, denoted by $U$, are unitary matrices representing transitions between quantum states. $\mathcal{U}(2^n)$ stands for the set of all $n$-qubit gates. Y rotations, $R_y$, are single-qubit unitaries $U\in\mathcal{U}(2)$ that redistribute the amplitude between $\ket{0}$ and $\ket{1}$. Z rotations, $R_z$, introduce the relative phase shift between $\ket{0}$ and $\ket{1}$. 
\begin{equation*}\label{eqn:u-ry}
    R_y(\theta) = \begin{pmatrix}
        \cos\frac{\theta}{2} & -\sin\frac{\theta}{2} \\
        \sin\frac{\theta}{2} & \cos\frac{\theta}{2}
    \end{pmatrix}\,,\, 
    R_z(\theta) = \begin{pmatrix}
        \text{exp}(-\frac{i\theta}{2}) & 0 \\
        0 & \text{exp}(\frac{i\theta}{2})
    \end{pmatrix},
\end{equation*}
where $\theta$ is the \emph{rotation angle}.

Pauli-X, $\sigma_x$, is a single-qubit operator that maps $\alpha\ket{0} + \beta\ket{1}$ to $\beta\ket{0} + \alpha\ket{1}$. CNOT is a two-qubit gate that transitions 
$\alpha\ket{00}\!+\!\beta\ket{01}\!+\!\gamma\ket{10}\!+\!\delta\ket{11}$ to 
$\alpha\ket{00}\!+\!\beta\ket{01}\!+\!\gamma\ket{11}\!+\!\delta\ket{10}$. 
\begin{equation*}\label{eqn:u-x}
    \sigma_x = \begin{pmatrix}
        0 & 1 \\
        1 & 0 
    \end{pmatrix},\ 
    \text{CNOT} = \mathbbm{1} \oplus \sigma_x = \begin{pmatrix}
        1 & 0 & 0 & 0 \\
        0 & 1 & 0 & 0 \\
        0 & 0 & 0 & 1 \\
        0 & 0 & 1 & 0 \\
    \end{pmatrix}.
\end{equation*}

All operators in $\mathcal{U}(2^n)$ can be decomposed into gates in the set $\{\text{CNOT},\  \mathcal{U}(2)\}$~\cite{barenco1995elementary}. A single-qubit gate $\mathcal{U}(2)$ can be further decomposed into $R_z$ and $R_y$ deterministically~\tocheck{\cite{krol2022efficient}}. 

Two operators, $U_1$ and $U_2$, \emph{commute} if exchanging their order does not affect the functionality, i.e., $U_1U_2 = U_2U_1$.

\subsection{Quantum State Preparation}
Given a state $\ket{\psi}$, the \emph{quantum state preparation}~(QSP) finds a quantum circuit comprising $l$ gates $U_1, U_2, ... U_l$ such that these gates transition the ground state $\ket{0}$ to the target state $\ket{\psi}$, i.e., $\ket{\psi} = U_l...U_2U_1\ket{0}$. For \emph{noisy intermediate-scale quantum}~(NISQ) computers, CNOTs produce far more noise than single-qubit gates. Besides, CNOTs require the two involved qubits to be coupled and generate connectivity constraints in the layout synthesis~\cite{iccad20-tan-cong-optimal-layout-synthesis}. Therefore, the objective of the QSP circuit optimization in this paper is to minimize the number of CNOTs in the circuit after being decomposed into gates in $\{\text{CNOT}, \mathcal{U}(2)\}$.

Additionally, this paper studies states with real amplitudes due to their predominance in quantum algorithms. To prepare arbitrary states with complex amplitudes, our method adapts through the controlled phase gates~\cite{Amy_2018}.


\section{Motivating Example}\label{sec:motivating-example}
In this section, we illustrate the limitations of applying conventional quantum compilation algorithms to optimize the quantum state preparation circuit. Then, we demonstrate the advantages of our optimization approach, which leverages don't-care conditions and resynthesis subcircuits to achieve a significant reduction in the CNOT count and the error rate.


Consider the target state $\ket{\psi}$. Each index encodes the value of three qubits in the order of $\ket{q_0q_1q_2}$. 
\begin{equation}\label{eqn:target-state}
\begin{array}{rl}
    \ket{\psi} = &\sqrt{\frac{2}{8}}\ket{000}-\sqrt{\frac{1}{8}}\ket{100}+\sqrt{\frac{1}{8}}\ket{010}+ \\
    &\sqrt{\frac{1}{8}}\ket{101}+\sqrt{\frac{1}{8}}\ket{011}+\sqrt{\frac{2}{8}}\ket{111}
\end{array}
\end{equation}

\Cref{fig:optimized-conventional} depicts a feasible initial QSP circuit of $\ket{\psi}$. This initial circuit comprises five segments, $\mathcal{W}_1$, ..., $\mathcal{W}_5$. These segments are generated by decomposing unitary operators, including two single-qubit Y rotations ($\mathcal{W}_1$ and $\mathcal{W}_4$), two single-controlled Y rotations ($\mathcal{W}_2$ and $\mathcal{W}_5$), and one double-controlled Y rotation ($\mathcal{W}_3$). For a rotation with $n$ control qubits, the configuration requires $2^n$ CNOTs after mapping to $\libsmall$~\cite{krol2022efficient}. Therefore, The final circuit needs $10$ single-qubit gates and $8$ CNOT gates to prepare $\ket{\psi}$. 
\begin{table}[t]
    \centering
    \caption{Intermediate quantum states in the example.}\label{tab:intermediate-states}
\vspace{-3mm}
    \begin{tabular}{ll}
    \hline
        $\ket{\psi_1}=$ & $\sqrt{\frac{1}{2}}\ket{000}+\sqrt{\frac{1}{2}}\ket{100}$\\
        $\ket{\psi_2}=$ & $\sqrt{\frac{2}{4}}\ket{000}+\sqrt{\frac{1}{4}}\ket{100}+\sqrt{\frac{1}{4}}\ket{110}$\\
        $\ket{\psi'_2}=$ & $\sqrt{\frac{1}{4}}\ket{000}+\sqrt{\frac{1}{4}}\ket{010}+\sqrt{\frac{1}{4}}\ket{100}+\sqrt{\frac{1}{4}}\ket{110}$\\
        $\ket{\psi_3}=$ & $\sqrt{\frac{2}{4}}\ket{000}+\sqrt{\frac{1}{4}}\ket{101}+\sqrt{\frac{1}{4}}\ket{110}$\\
        $\ket{\psi_4}=$  & $\sqrt{\frac{2}{8}}\ket{000}-\sqrt{\frac{1}{8}}\ket{100}+\sqrt{\frac{1}{8}}\ket{111}+\sqrt{\frac{1}{8}}\ket{101}+$ \\
        &$\sqrt{\frac{1}{8}}\ket{110}+\sqrt{\frac{2}{8}}\ket{010}$ \\
        $\ket{\psi_5}=$  & $\sqrt{\frac{2}{8}}\ket{000}-\sqrt{\frac{1}{8}}\ket{100}+\sqrt{\frac{1}{8}}\ket{010}+\sqrt{\frac{1}{8}}\ket{101}+$ \\
        &$\sqrt{\frac{1}{8}}\ket{011}+\sqrt{\frac{2}{8}}\ket{111}$ \\
    \hline
    \end{tabular}
\end{table}

Existing optimization algorithms rigidly preserve the untary equivalency of the quantum circuits~\cite{Qiskit, younis2020qfast, smith2023leap}. Let $U$ and $U'$ be the unitaries of the initial and optimized circuit. These algorithms assert the optimized circuit must behave the same as the initial circuit for all possible states in the three-qubit system, i.e.,
\begin{equation}\label{eqn:unitary-without-dc}
    U\varphi = U'\varphi\,,\, \forall \varphi\in\mathcal{H}^{\otimes 3}.
\end{equation}

This constraint is over-conservative for QSP circuit optimization and hinders conventional algorithms from finding more effective circuits with fewer gates. \Cref{fig:satisfiability-dc} displays a smaller circuit that prepares $\ket{\psi}$ using $5$ CNOT gates. Let $U_{w_i}$ and $U'_{{w_i}}$ be the unitaries of the $i$th initial and optimized segments $\mathcal{W}_i$, respectively. Although the unitaries of their subcircuits are different, e.g., $U_{w_2} \neq U'_{w_2}$, the outcome states are the same, i.e., $U_{w_2}\ket{\psi_1} = U'_{w_2}\ket{\psi_1} = \ket{\psi_{2}}.$ It can be observed by expressing the matrix representation in terms of the subspace generated by $q_0$ and $q_1$ ($q_2$ is separable in $\mathcal{W}_2$'s state transition). 
\begin{equation*}
\setlength{\arraycolsep}{3pt}
    U_{w_2} = \begin{pmatrix}
        1 & 0 & 0 & 0 \\
        0 & 1 & 0 & 0 \\
        0 & 0 & 0.7 & -0.7 \\
        0 & 0 & 0.7 & 0.7 \\
    \end{pmatrix},\, U'_{w_2} = \begin{pmatrix}
        1 & 0 & 0 & 0 \\
        0 & \!-\!1 & 0 & 0 \\
        0 & 0 & 0.7 & -0.7 \\
        0 & 0 & 0.7 & 0.7 \\
    \end{pmatrix},\, \ket{\psi_1} = \begin{pmatrix}
        0.7 \\
        0 \\
        0.7 \\
        0 \\
    \end{pmatrix}.
\end{equation*}

The constraint in \Cref{eqn:unitary-without-dc} allows the optimization algorithm to modify the functionality and simplify the circuit more aggressively than \Cref{fig:optimized-conventional}, which reduces the gate count. 

Nevertheless, the optimization in \Cref{fig:satisfiability-dc} is again conservative because preserving the equivalency of intermediate states is also unnecessary. An even more aggressive optimization is revealed in \Cref{fig:observability-dc}.  Since a circuit $U'$ is feasible if the outcome of the overall circuit is the target state $\ket{\psi}$, i.e., $U\ket{0} = U'\ket{0} = \ket{\psi}$. Therefore, this algorithm allows altering the intermediate state from $\ket{\psi_2}$ to $\ket{\psi'_2}$. This modification facilitates a significantly simplified circuit configuration to prepare $\ket{\psi}$ employing only $2$ CNOT gates. 
\begin{table}[b]
\vspace{-2mm}
\caption[]{Error rate compared to the ideal value, i.e., $|c_x|^2$, when running on the quantum device \texttt{ibm-osaka}. We display the actual probability and the absolute error (in parenthesis). Optimizing the QSP circuit using don't cares results in better state preparation fidelity\footnotemark.}
\label{tab:qiskit-experiment}
\centering
\begin{tabular}{c|c|c|c|c}
\hline
 & Ideal          & Initial      & Opt. w/o DC  & Opt. with DC      \\
\hline
$\ket{000}$   & 0.250  &  0.233 (1.7\%) & 0.226 (2.4\%) & 0.265 (1.5\%) \\
$\ket{001}$   & 0.125  &  0.115 (0.9\%) & 0.116 (0.8\%) & 0.119 (0.5\%) \\
$\ket{010}$   & 0.125  &  0.093 (3.2\%) & 0.223 (9.8\%) & 0.105 (2.0\%) \\
$\ket{011}$   & 0.000  &  0.039 (3.8\%) & 0.022 (2.2\%) & 0.007 (0.7\%) \\
$\ket{100}$   & 0.000  &  0.023 (2.3\%) & 0.013 (1.3\%) & 0.004 (0.3\%) \\
$\ket{101}$   & 0.125  &  0.230 (10.\%) & 0.074 (5.0\%) & 0.146 (2.0\%) \\
$\ket{110}$   & 0.125  &  0.061 (6.4\%) & 0.167 (4.2\%) & 0.130 (0.4\%) \\
$\ket{111}$   & 0.250  &  0.206 (4.4\%) & 0.158 (9.2\%) & 0.224 (2.5\%) \\
\hline
\multicolumn{2}{c|}{Avg. error} & 4.18\%        & 4.40\%        & 1.29\%         \\
\hline
\end{tabular}
\vspace{-2mm}
\end{table}

\Cref{tab:qiskit-experiment} shows the performance on a real quantum computer of the initial circuit in \Cref{fig:optimized-conventional}, circuit optimized by Qiskit~\cite{Qiskit}, and circuit in \Cref{fig:observability-dc}. Results show that employing don't-care conditions allows the optimized circuit to prepare the same state with higher fidelity and reduce the average error compared to the ideal state from $4.18\%$ to $1.29\%$. 

\footnotetext{We transpile all three circuits using $\texttt{qiskit}$ with an optimization level of $3$ and run $4096$ shots for each experiment.}

\section{Methodology}\label{sec:methodology}
As established in \Cref{sec:motivating-example}, allowing modifications in the unitary matrices significantly enhances the optimization flexibility of the QSP circuit. However, this introduced flexibility also expands the solution space, potentially increasing the complexity if not managed properly. The main contribution of this paper is the development of a scalable algorithm that exploits these don't care conditions effectively while keeping the complexity under control.

Our quantum circuit optimization workflow is based on the peephole optimization, as presented in Algorithm~\ref{alg:dc-resyn}. Starting with an initial QSP circuit $\mathcal{G}$, we sequentially traverse the gates and extract segments that target the same qubit. For each extracted segment, denoted as $\mathcal{W}$, we perform \emph{resynthesis} that searches for a replacement circuit $\mathcal{W}'$ that transitions the quantum states identically while requiring fewer gates. The optimized circuit is then produced by aggregating these improvements across all segments. 

In the rest of this section, we will explain the three main components of our algorithm: segment extraction, rotation table representation, don't care condition propagation, and segment resynthesis. 

\subsection{Single-Target Segment Extraction}\label{subsec:single-target-window}

Segment extraction is a widely used technique to simplify the complexity of circuit optimization. By focusing the optimization process on a localized segment of the quantum circuit, we can represent the unitary functions more efficiently and significantly reduce the problem size. However, local optimization lacks the comprehensive visibility of feasible solutions, resulting in missed optimization opportunities. Therefore, it is crucial for a partitioning algorithm to strike a balance between efficiency and effectiveness.
\begin{algorithm}[bt]
\caption{QSP Circuit Resynthesis}\label{alg:dc-resyn}
\SetKwInOut{Input}{input}
\SetKwInOut{Output}{output}

\DontPrintSemicolon
\Input{An initial QSP circuit $\mathcal{G}$ comprises a sequence of quantum gates, $U_1$, $U_2$, ..., $U_l$, where $U_i\in \libsmall$.}
\Output{An optimized circuit $\mathcal{G}'$.}
$W \gets \varnothing$ \;
$\ket{\varphi_t} \gets \ket{0}$ \;
\For{$i = 1, 2, ..., l$} {
    $\ket{\varphi_s} \gets \ket{\varphi_t}$ \;
    $U_{ext} \gets \mathbbm{1}$ \;
    $\mathcal{W} \gets \varnothing$ \;
    \For{$j = i, ..., l$} {
        \If{$U_{j}.\text{target} = U_{i}.\text{target}$ and $\text{isCommute}(U_{j}, U_{ext})$} {
            $\mathcal{W}.\text{push}(U_{j})$ \;
            $\ket{\varphi_t} = U_j\ket{\varphi_t}$ \;
        }
        \Else {
            \If {$U_{j}.\text{control} = U_i.\text{target}$} {
                \textbf{break} \;
            }
            $U_{ext} \gets U_{j}U_{ext}$
        }
    }
    $W$.push($\mathcal{W}$)\;
}
propagateDontCares($W$) \;
$\mathcal{G}' \gets \varnothing$ \;
\For{$\mathcal{W} \in W$} {
    $\mathcal{W}' \gets \text{resynthesis}(\mathcal{W}, \ket{\varphi_s}, \ket{\varphi_t})$ \;
    $\mathcal{G}' \gets \mathcal{G}' \cup \mathcal{W}'$ \;
}
\Return $\mathcal{G}'$.\;
\end{algorithm}

Our algorithm extracts single-target segments. A segment is \emph{single-target} if all CNOTs and single-qubit gates in it target the same qubit. Let $U_w$ be the unitary of the segment $\mathcal{W}$, then $\mathcal{W}$ is single-target if and only if $U_w$ can be represented as a multi-controlled single-target operator
\begin{equation}\label{eqn:single-target}
    U_w = \sum_{x\in \{0,1\}^{n-1}}U_{w, x}\otimes \ket{x}\bra{x},
\end{equation}
where $U_x\in \mathcal{U}(2)$ is unitary matrix that indicates the effect $\mathcal{W}$ casts on the target qubit if the other $n-1$ qubits are in state $\ket{x}$. 


To extract single-target segments, we first group adjacent gates with common targets in the original circuit sequence and then extend the boundary of each segment utilizing commute operators. As shown in Algorithm~\ref{alg:dc-resyn} from lines 8 to 14, we traverse all the gates $U_j$ sequentially after $i$ and collect as many gates with the common target ($U_i$.target) to $\mathcal{W}$ as possible. To achieve this, we keep track of the operator $U_{ext}$, which represents the traversed unitaries with different target qubits and excluded from $\mathcal{W}$. A gate $U_j$ with the same target qubit can be reordered before $U_{ext}$ and inserted into the segment $\mathcal{W}$ if $U_j$ and $U_{ext}$ commute. Therefore, by allowing the reordering of commutable operators in the sequence, we can gather more gates into each single-target segment, thereby increasing the potential for optimizing gate count. 

Observe that we can terminate the check for commute operators if the target qubit $t=U_i.target$ serves as a control qubit in $U_{ext}$ (line \tocheck{12}). In this case, $U_{ext}$ is not likely to commute with any $U_{j}$ that targets $t$, and we exit the loop to accelerate the algorithm.

\subsection{Rotation Table Representation}\label{subsec:rotation-table}
The advantage of single-target segment partitioning is that the local unitaries can be expressed efficiently. Instead of complex matrix multiplications, the functionality can be represented using additions and subtraction of rotation angles. This simplification will noticeably enhance the efficiency of simulating state transitions in the segment and facilitates the implementation of resynthesis algorithms. 

\begin{theorem}\label{theorem-single-target}
    Given an arbitrary state $\psi_s$ with real amplitudes and a single-target segment $\mathcal{W}$ comprises real unitaries from $\{\mathcal{U}(2)$, CNOT$\}$, there exists a unitary operator $U'_w$, 
    \begin{equation}
        U'_w = \sum_{x\in \{0,1\}^{n-1}}R_y(\theta_{w,x})\otimes \ket{x}\bra{x},
    \end{equation}
    where $-2\pi \leq \theta_{w,x} < 2\pi$, such that $U'_w\ket{\psi_s} = U_w\ket{\psi_s}$. 
\end{theorem}
\begin{proof}
    We give a proof based on induction. Assume this property holds true for all segments with no more than $k$ gates, $k\geq 1$. Consider a single-target segment $\mathcal{W}$ with $k+1$ gates. We can partition it into two single-target segment $\mathcal{W}_1$ and $\mathcal{W}_2$ with gate counts $k_1, k_2 \leq k$, such that there exists:
    \begin{equation}
\arraycolsep=1.pt\def\arraystretch{1}
    \begin{array}{rl}
         U_{w}\ket{\psi_s} &= U_{w_1}U_{w_2}\ket{\psi_s} \\
         &= (\sum_{x}R_{y}(\theta_{w_1,x})\otimes \ket{x}\bra{x})(\sum_{x}R_y(\theta_{w_2,x})\otimes \ket{x}\bra{x})\ket{\psi_s} \\
         &= (\sum_{x}R_y(\theta_{w_1,x}+\theta_{w_2,x})\otimes \ket{x}\bra{x})\ket{\psi_s},
    \end{array}
    \end{equation}
    which implies that the property also applies to segments with $k+1$ gates. 

    To complete the proof, it suffices to show that a single $\mathcal{U}(2)$ or CNOT can be replaced by an $R_y$ operator. It holds true because the initial and final state, $\ket{\psi_s}$ and $U_{\mathcal{W}}\ket{\psi_s}$ are given. Let $\alpha_{s,x}\ket{0} + \beta_{s,x}\ket{1}$ and $\alpha_{t,x}\ket{0} + \beta_{t,x}\ket{1}$ be the initial and final states of the target qubit for a control state $\ket{x}$. Then this transition can be accomplished by a Y rotation of $\theta_x$, 
    \begin{equation}
        \theta_x = 2\cdot\text{atan2}(\beta_{t,x},\alpha_{t,x}) - 2\cdot\text{atan2}(\beta_{s,x},\alpha_{s,x}),
    \end{equation}
    where $\text{atan2}$ is applied instead of $\arctan$ to determine the correct quadrant of the angle. Finally, we can find $\theta_x$ to replace $\mathcal{U}(2)$ and CNOTs given the initial state $\psi_s$. 
\end{proof}

We use \emph{rotation tables} to indicate the mapping between the control qubit states, $\ket{x}$, and the corresponding rotation angles, $\theta_x$. This mapping is formalized as a function $\varphi:\{0,1\}^{n-1}\rightarrow [-2\pi,2\pi)$. 
\Cref{tab:rotation-table} displays the rotation table of states listed in \Cref{tab:intermediate-states} with respect to qubit $q_2$. The first column represents the states of the control qubits, and the entries are rotation angles ranging from $-2\pi$ to $2\pi$. For example, the angle of $\ket{10*}$ in $\ket{\psi_5}$ corresponds to two terms, $\ket{100}$ and $\ket{101}$, with coefficients $\alpha=-\sqrt{\frac{1}{8}}$ and $\beta=\sqrt{\frac{1}{8}}$. Therefore, the corresponding rotation angle is $2\cdot\text{atan2}(\sqrt{\frac{1}{8}}, -\sqrt{\frac{1}{8}}) = \frac{3\pi}{2}$. 
\begin{table}[h]
\normalsize
    \centering

    \caption{Rotation tables of $q_2$ correspond to the state transition from $\ket{\psi_2'}$ to $\ket{\psi_5}$ in \Cref{fig:observability-dc}. $\ket{\psi_2'}\!=\!\sqrt{\frac{1}{4}}\ket{000}\!+\!\sqrt{\frac{1}{4}}\ket{010}\!+\!\sqrt{\frac{1}{4}}\ket{100}\!+\!\sqrt{\frac{1}{4}}\ket{110}$ and $\ket{\psi_5}\!=\!\sqrt{\frac{2}{8}}\ket{000}\!-\!\sqrt{\frac{1}{8}}\ket{100}\!+\!\sqrt{\frac{1}{8}}\ket{010}\!+\!\sqrt{\frac{1}{8}}\ket{101}\!+\!\sqrt{\frac{1}{8}}\ket{011}\!+\!\sqrt{\frac{2}{8}}\ket{111}$.}\label{tab:rotation-table}
\vspace{-2mm}
    \begin{tabular}{l|ccc}
    \hline
        Control $\ket{x}$ & $\ket{\psi'_2}$ & $\ket{\psi_5}$ & $\theta_x$\\
    \hline
        $\ket{00*}$ & $0$  & $0$      & $0$      \\
        $\ket{01*}$ & $0$  & $\pi/2$  & $\pi/2$  \\
        $\ket{10*}$ & $0$  & $3\pi/2$ & $3\pi/2$ \\
        $\ket{11*}$ & $0$  & $\pi$    & $\pi$    \\
    \hline
    \end{tabular}
\vspace{-3mm}
\end{table}


\subsection{Don't Care Condition Propagation}\label{subsec:dont-care-propagation}
After partitioning the circuit into segments, line 16 in Algorithm~\ref{alg:dc-resyn} identifies the conditions within each segment where modifying its functionality does not impact the final outcome. The derivation and exploitation of these \emph{don't-care} conditions are well-established in the field of logic synthesis, where they significantly contribute to the simplification of Boolean logic circuits~\cite{savoj1991use, bartlett1988multi, damiani1990observability}. This paper adheres to the established terminologies in logic synthesis, extends the concept of don't-cares, and applies Boolean methods to optimize quantum circuits.  

All the don't care conditions utilized in this paper originate from the incompletely specified mappings between basis states. Let $V=\{v_1, ..., v_{2^n}\}$ and $W=\{w_1, ..., w_{2^n}\}$ be two sets of orthonormal basis of $\mathcal{H}^{\otimes n}$. Given $m$ specified mappings, $m \in \{1, ..., 2^n\}$, the circuit unitary can be expressed as:  
\begin{equation}\label{eqn:general-exdc}
    U = \sum_{i\leq m}\ket{w_i}\bra{v_i} + \sum_{i> m}\ket{w_{f(i)}}\bra{v_i},
\end{equation}
where $f$ randomly permutes the mappings of the rest $2^n-m$ vectors between $V$ and $W$. Szasz et al. demonstrate that this formulation generalizes the flexibilities in various quantum applications, including Hamiltonian simulation circuit synthesis and quantum channel preparation~\cite{szasz2023numerical}. 

We define \emph{external don't cares} of operator $U$ as the set of vectors, $\{v_{m+1}, ..., v_{2^n}\}$, with unspecified final states after applying $U$, and we define the corresponding \emph{care set} as the set $C_{\textit{ext}} = \{v_1, ..., v_m\}$. Since the QSP circuit usually transitions the quantum states from the ground state $\ket{0}$ to the target state $\ket{\psi}$, the external care set $C_{\textit{ext}} = \{\ket{0}\}$ contains only the ground state. All other $2^n-1$ basis vectors $x\in\{0,1\}^n \setminus C_{\textit{ext}}$ are considered external don't cares. 

To utilize don't cares in our scalable workflow, we need to propagate external don't cares to the single-target segments. According to their properties, we categorized the don't-cares of a segment into \emph{controllability don't cares}~(CDC) and \emph{observability don't cares}~(ODC).

\begin{table}[b]
\normalsize
    \centering
    \caption{Rotation tables of $q_1$ corresponds to the state transitions from $\ket{\psi_1}$ to $\ket{\psi_2}$ and $\ket{\psi'_2}$ in \Cref{tab:intermediate-states}. ``X'' represents controllability don't cares and ``$X^*$'' is the observability don't care}\label{tab:rotation-table-q1}
\vspace{-3mm}
    \begin{tabular}{l|ccc}
    \hline
        Control state & $\ket{\psi_1}$ & $\ket{\psi_2}$ & $\ket{\psi'_2}$  \\
    \hline
        $\ket{0*0}$ & $0$ & $0$       & $X^*$     \\
        $\ket{0*1}$ & $X$ & $X$       & $X$     \\
        $\ket{1*0}$ & $0$ & $\pi/2$   & $\pi/2$   \\
        $\ket{1*1}$ & $X$ & $X$       & $X$     \\
    \hline
    \end{tabular}
\vspace{-3mm}
\end{table}

\subsubsection{Controllability Don't Cares (CDC)}
Controllability don't care is the set of states that will not occur (i.e., have zero amplitude) when entering a segment $\mathcal{W}$. Consequently, the behavior of the target qubits corresponding to non-existent control states does not affect the outcome. Let $\psi_s$ represent the initial state of the segment $\mathcal{W}$, and $U_s$ be the unitary of the circuits prior to $\mathcal{W}$, $\ket{\psi_s} = U_s\ket{0}$. Then, the care set of function $U_{\mathcal{W}}$ is:
\begin{equation}\label{eqn:care-set-cdc}
    C_{\mathcal{W}} = \left\{\ket{x}\otimes \left(\cos(\frac{\varphi_{s,x}}{2})\ket{0}+\sin(\frac{\varphi_{s,x}}{2})\ket{1}\right):x\in S(\psi_s)\right\},    
\end{equation}
where $\varphi_{s,x}$ represents the rotation angle corresponds to control state $\ket{x}$ and $S(\psi_s)$ is the index set of $\psi_s$ that consists of the basis vectors with non-zero amplitudes. 

CDC can be efficiently captured by simulating the state transitions from the input ground state on the fly, as displayed in lines~\tocheck{4 and 10} in Algorithm~\ref{alg:dc-resyn}. For instance, the segment $\mathcal{W}_{2}$ in \Cref{fig:optimized-conventional} targeting $q_1$ transition the states from $\psi_1$ to $\psi_2$, whose rotation tables can be derived during simulation. The control states $\ket{0\!*\!1}$ and $\ket{1\!*\!1}$ do not exist in both $\ket{\psi_1}=\sqrt{\frac{1}{2}}\ket{000}+\sqrt{\frac{1}{2}}\ket{100}$ and $\ket{\psi_2} = \sqrt{\frac{2}{4}}\ket{000}+\sqrt{\frac{1}{4}}\ket{100}+\sqrt{\frac{1}{4}}\ket{110}$, thus, are CDC. The care set, according to \Cref{eqn:care-set-cdc} is $\{\ket{000}, \ket{100}\}$ and has a cardinality $2$ out of all $2^3=8$ basis vectors, as shown in \Cref{tab:rotation-table-q1}. 

Note that the cardinality of the care set $C_{\mathcal{W}}$ is upper-bounded:
\begin{equation*}
    |C_{\mathcal{W}}| \leq \min(2^{n-1}, |S(\psi_s)|),
\end{equation*}
where $n$ is the number of qubits and $|S(\psi_s)|$ represents the cardinality of the input state. This is because we need to specify at most one mapping for each control state $\ket{x}$, $x\leq 2^{n-1}$. If the control state $\ket{x}$ does not exist in $S(\psi_s)$, the outcome vector is in the don't care set. 

\subsubsection{Observability Don't Cares (ODC)} 
Observability don't care is the set of control states whose output state that does not affect the final outcome of the QSP circuit. Let $\psi_i$ be the final state of a segment $\mathcal{W}_i$ targets $q_i$, $U_{j}$ be the initial unitary of the next segment $\mathcal{W}_{j}$ targets $q_j$, $j \neq i$. Then the rotation angle of a control state $\ket{x}\otimes \ket{*}_i$ is in $\mathcal{W}_i$'s ODC if the projections of $U_j$ on $\ket{x}\otimes \ket{0}_i$ and $\ket{x}\otimes \ket{1}_i$ are the same, i.e.:
\begin{equation}
    \bra{0}_i\bra{x}U_{j}\ket{x}\ket{0}_i = \bra{1}_i\bra{x}U_{j}\ket{x}\ket{1}_i.
\end{equation}

Unlike CDC, which are derived from the predecessors of a segment, observability don't cares are generated and propagated from the opposite direction, from the successors. For example, the rotation angle for the control state $\ket{0\!*\!0}$ in \Cref{tab:rotation-table-q1} is an observability don't care of $U_{w_2}$ generated by $\mathcal{W}_3$ from \Cref{fig:optimized-conventional}.
Note that the eight gates in $\mathcal{W}_3$ compose a controlled Y rotation that is activated only when $q_0$ is $\ket{1}$ and $q_1$ is $\ket{0}$. 
Therefore, when $q_0$ is $\ket{0}$, $U_{w_3}$ does not modify $q_2$ regardless of the value of $q_1$. In other words, $\ket{000}$ and $\ket{010}$ are equivalent for $\mathcal{W}_2$. Consequently, although the control state $\ket{0\!*\!0}$ has non-zero amplitude, its rotation angle does not affect the feasibility of the circuit and, thus, is ODC. 

Exploiting ODC can further simplify the circuits in the segment. As illustrated in \Cref{tab:rotation-table-q1}, after setting the angle in the first row of $\ket{\psi'_2}$ from $0$ to $\frac{\pi}{2}$, the original segment can be replaced by one single $R_y$ gate. After traversing and optimizing the entire circuit, we derive the circuit in \Cref{fig:observability-dc} from \Cref{fig:optimized-conventional} utilizing both CDC and ODC.

It is worth mentioning that although the rotation of the segment $\mathcal{W}_3$ is not affected in the previous example, changing the angle from $0$ to $\frac{\pi}{2}$ in \Cref{tab:rotation-table-q1} produces unintended rotations on $q_1$. To mitigate the cost of reverting these rotations after $\mathcal{W}_j$, we only exploit observability don't care if the segment $\mathcal{W}_{j+1}$ after $\mathcal{W}_j$ targets the same qubit as $\mathcal{W}_{j-1}$, which introduced unintended rotation. In this case, we can simply adjust the rotation table of $\mathcal{W}_{j+1}$ to completely recover the functionality, which does not necessarily increase the gate count of $\mathcal{W}_{j+1}$.

\subsection{Segment Resynthesis}
As displayed in line~\tocheck{19} in Algorithm~\ref{alg:dc-resyn}, the segment resynthesis is the final step of our method that finds the replacement of the initial segment with fewer gates that accomplishes the same state transition from $\psi_s$ to $\psi_t$. This step serves as the core of our workflow, significantly decreasing gate count and optimizing the overall circuit.

The workflow of the resynthesis algorithm is illustrated in Algorithm~\ref{alg:window-resyn}. Given a target rotation table, our algorithm:
\begin{enumerate}
    \item defines a hyper-parameter $K$ as the number of CNOTs in the template and initializes it to $K=0$. 
    \item constructs an equality system to check if the rotation table can be implemented if using $K$ CNOTs.
    \item analyzes the result. If a set of solutions is found, we return the circuit with $K$ CNOTs. Otherwise, we increase $K$ by one and repeat step (2).
\end{enumerate}

\begin{algorithm}[tb]
\caption{Segment resynthesis}\label{alg:window-resyn}
\SetKwInOut{Input}{input}
\SetKwInOut{Output}{output}

 \DontPrintSemicolon
 \Input{The initial and final rotation tables $\varphi_s$ and $\varphi_t$ with $n$ control qubits.}
 \Output{A sequence of quantum operators to prepare the rotation table with $K$ CNOTs.}
 
$K \gets 0$ \;
$\text{solution} \gets \varnothing$ \;
\While{$\text{solution}$ is $\varnothing$ and $K \leq K_{\max}$} {
    \For{$\{S_k\}$ in {$\mathbb{Z}_{n}^{\otimes K}$}} {
        $M \gets $ Construct equality system ($\varphi_s$, $\varphi_t$, $\{S_k\}$)\;
        solution $\gets $ Solve($M$) \;
    }
    $K \gets K+1$ \;
}
\Return solution.\;
\end{algorithm}

All utilized variables are listed in \Cref{tab:variable-declaration-mcry-decomposition}. In the rest of this section, we will demonstrate the construction of the equality system from line~\tocheck{5} in Algorithm~\ref{alg:window-resyn}.
\begin{table}[b]
\small
    \centering
    \caption{\small Variable declaration for the segment resynthesis problem.}\label{tab:variable-declaration-mcry-decomposition}

\setlength{\tabcolsep}{2pt}
    \begin{tabular}{lll}
    \hline
    \multicolumn{3}{l}{\textbf{Input:}} \\
    $S_{k,i}$ & $\{0,1\}$ & Whether $q_i$ is the control qubit in the $k$th CNOT.\\
    $R_{k,x}$ & $\{-1,1\}$ & $R_{x,k}\!=\!-1$ if $\ket{x}$ is activated by the $k$th CNOT. \\
    $\varphi_{s,x},\varphi_{t,x}$ & $[-2\pi, 2\pi)$ & Initial and final angles corresponds to $\ket{x}$. \\
    \multicolumn{3}{l}{\textbf{Internal:}} \\
    $\varphi_{k,x}$ & $[-2\pi, 2\pi)$ & The angle corresponds to $\ket{x}$ after the $k$th CNOT. \\
    \multicolumn{3}{l}{\textbf{Output:}} \\
    $\theta_k$ & $[0, 4\pi)$ & The rotation angle of R$_y$ after the $k$th CNOT. \\
    \hline
    \end{tabular}
\end{table}

\subsubsection{Fixed-target CNOT template}

Since all the CNOT gates in the circuit target the same qubit, the control qubit alone suffices to represent each CNOT gate. We define CNOT variables, $S_{k,i}$, as binary values indicating whether $q_i$ serves as the control qubit in $k$th CNOT. Given that each CNOT has only one control qubit, $\sum_{i\in \mathbb{Z}_{n}}S_{k,i} = 1$ for all $k\in \mathbb{Z}_{K},$ where $n$ represents the number of candidate control qubits. Then, we define a set of variables $R_{k,x}$, where $R_{k,x} = -1$ if the rotation at index $\ket{x}$ is enabled by the $k$th CNOT, and $R_{k,x} = 1$ otherwise. These $R_{k,x}$ variables are derived from the given $S_{k,i}$ variables, as the control qubit $i$ and its phase determines whether an index $x$ is affected by the CNOT gate. 
\begin{figure}[hbt]
\centering
\mbox{
\small
    \Qcircuit @C=.7em @R=.5em {
\lstick{q_0} & \qw  & \qw & \qw & \ctrl{2} & \qw  & \qw \\
\lstick{q_1} & \qw & \ctrl{1} & \qw & \qw & \qw &\qw \\
\lstick{q_2} & \gate{\parbox{0.6cm}{\centering \footnotesize $R_y$\\$\theta_0$}} & \targ & \gate{\parbox{0.6cm}{\centering \footnotesize $R_y$\\$\theta_1$}} & \targ & \gate{\parbox{0.6cm}{\centering \footnotesize $R_y$\\$\theta_2$}} & \qw
    }
}
\vspace{-2mm}
\caption{CNOT template with $K=2$, $S_{1,1}=1$ and $S_{2,0}=1$.}\label{fig:cnot-template}
\end{figure}
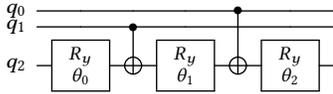

\subsubsection{Rotation angle variables and constraints}

We use $\varphi_{k,x}$ to represent the rotation angle corresponding to index $\ket{x}$ after the $k$th CNOT gate. The values of the initial and final angles, $\varphi_{0,x} = \varphi_{s,x}$ and $\varphi_{K,x} = \varphi_{t,x}$, are given by the rotation tables. Then, we define rotation angle constraints as shown in \Cref{eqn:rotation-angle-constraints}.
\begin{equation}\label{eqn:rotation-angle-constraints}
    \varphi_{k+1,x} = \frac{\pi}{2} + \theta_k + R_{k,x}\cdot(\varphi_{k,x}-\frac{\pi}{2}). 
\end{equation}
This constraint describes the evolution of the rotation table before and after the $k$th CNOT gate. If the index $\ket{x}$ is not enabled by the CNOT, then $R_{k,x} = 1$, and $\varphi_{k+1,x} = \varphi_{k,x} + \theta_k$ because the $R_y$ gate after CNOT introduces a rotation of $\theta_k$. Otherwise, if $\ket{x}$ is enabled, $R_{k,x} = -1$, then $\varphi_{k+1,x} = \pi - \varphi_{k,x} + \theta_k$ because the CNOT gate ``reflects'' the previous angle $\varphi_{k,x}$ along $\varphi = \frac{\pi}{2}$, as illustrated in the proof of Theorem~\ref{theorem-single-target}, before applying the rotation of $\theta_k$.



\begin{example}\label{example:segment-synthesis}
    Consider the segment between $\ket{\psi'_2}$ and $\ket{\psi_5}$, whose initial and final angles are provided in \Cref{tab:rotation-table}. Assume we are at the iteration where $K=2$, then we can assign the boundary values to the rotation angle variables:
    \begin{align*}
        \varphi_{0,\ket{00}} = \varphi_{0,\ket{01}} = \varphi_{0,\ket{10}} = \varphi_{0,\ket{11}} = \theta_0\,,\, \\ 
        \varphi_{2,\ket{00}} = 0,\,\varphi_{2,\ket{01}} = \frac{\pi}{2},\,\varphi_{2,\ket{10}} = \frac{3\pi}{2} ,\,\varphi_{2,\ket{11}} = \pi, 
    \end{align*}
    where $\theta_0$ is the offset introduced by the first $R_y$ gate. 

    The target is to find an assignment to $\theta_0$, $\theta_1$, $\theta_2$, $S_{1,i}$, and $S_{i,2}$, such that the equality system is satisfied. For clarity, we assume the control qubits of the two CNOTs are $S_{1,1}\!=\!1$ and $S_{2,0}\!=\!1$. This corresponds to the template in \Cref{fig:cnot-template}, where $R_{1,\ket{00}}\!=\!R_{1,\ket{01}}\!=\!R_{2,\ket{00}}\!=\!R_{2,\ket{01}}\!=\!1$ and $R_{1,\ket{10}}\!=\!R_{2,\ket{01}}=R_{1,\ket{11}}=R_{2,\ket{11}}=-1$. We can streamline the equality system by plugging in the known variables and canceling intermediate variables corresponding to the same $x$ as shown in \Cref{eqn:equality-system-example}.
    \begin{equation}\label{eqn:equality-system-example}
\setlength{\arraycolsep}{1pt}
        \left\{
        \begin{array}{rl}
            \varphi_{1,\ket{00}}&=\varphi_{0,\ket{00}} + \theta_1  \\
            \varphi_{1,\ket{01}}&=\pi - \varphi_{0,\ket{01}} + \theta_1  \\
            \varphi_{1,\ket{10}}&=\varphi_{0,\ket{10}} + \theta_1  \\
            \varphi_{1,\ket{11}}&=\pi - \varphi_{0,\ket{11}} + \theta_1  \\
            \varphi_{2,\ket{00}}&=\varphi_{1,\ket{00}} + \theta_2  \\
            \varphi_{2,\ket{01}}&=\varphi_{1,\ket{01}} + \theta_2  \\
            \varphi_{2,\ket{10}}&=\pi - \varphi_{1,\ket{10}} + \theta_2  \\
            \varphi_{2,\ket{11}}&=\pi - \varphi_{1,\ket{11}} + \theta_2  \\
        \end{array}
        \right. \rightarrow
        \left\{
        \begin{array}{rl}
            0&=\theta_0 + \theta_1 + \theta_2  \\
            \frac{\pi}{2}&=(\pi - \theta_0 + \theta_1) + \theta_2  \\
            \frac{3\pi}{2}&=\pi - (\theta_0 + \theta_1) + \theta_2  \\
            \pi &=\theta_0  - \theta_1 + \theta_2  \\
        \end{array}
        \right..
    \end{equation}

    The solution $\theta_0 = \theta_2 = \frac{\pi}{4}$, $\theta_1=-\frac{\pi}{2}$ is feasible. Therefore, we can find a circuit using $K=2$ CNOTs to implement this rotation table as seen in \Cref{fig:observability-dc}.

    We can verify $K=1$ has no feasible solution by watching the corresponding equality system in \Cref{eqn:equality-system-no-solution}. 
    \begin{equation}\label{eqn:equality-system-no-solution}
\setlength{\arraycolsep}{1pt}
        \left\{
        \begin{array}{rl}
            \varphi_{1,\ket{00}} &=\varphi_{0,\ket{00}} + \theta_1  \\
            \varphi_{1,\ket{01}} &=\varphi_{0,\ket{01}} + \theta_1  \\
            \varphi_{1,\ket{10}} &=\pi - \varphi_{0,\ket{10}} + \theta_1  \\
            \varphi_{1,\ket{11}} &=\pi - \varphi_{0,\ket{11}} + \theta_1  \\
        \end{array}
        \right. \rightarrow
        \left\{
        \begin{array}{rl}
            0&=\theta_0 + \theta_1  \\
            -\frac{\pi}{2}&=\theta_0 + \theta_1 \\
            \frac{\pi}{2}&=\pi - \theta_0 + \theta_1 \\
            \pi&=\pi - \theta_0 + \theta_1 \\
        \end{array}
        \right..
    \end{equation}

    Similarly, $K=0$ is also infeasible, and we return the 2-CNOT circuit as the optimal solution for this segment synthesis problem. 
\end{example}

\begin{figure*}[bt]
    \centering
    \begin{subfigure}[b]{.24\linewidth}
    \centering
        \includegraphics[width=.99\linewidth]{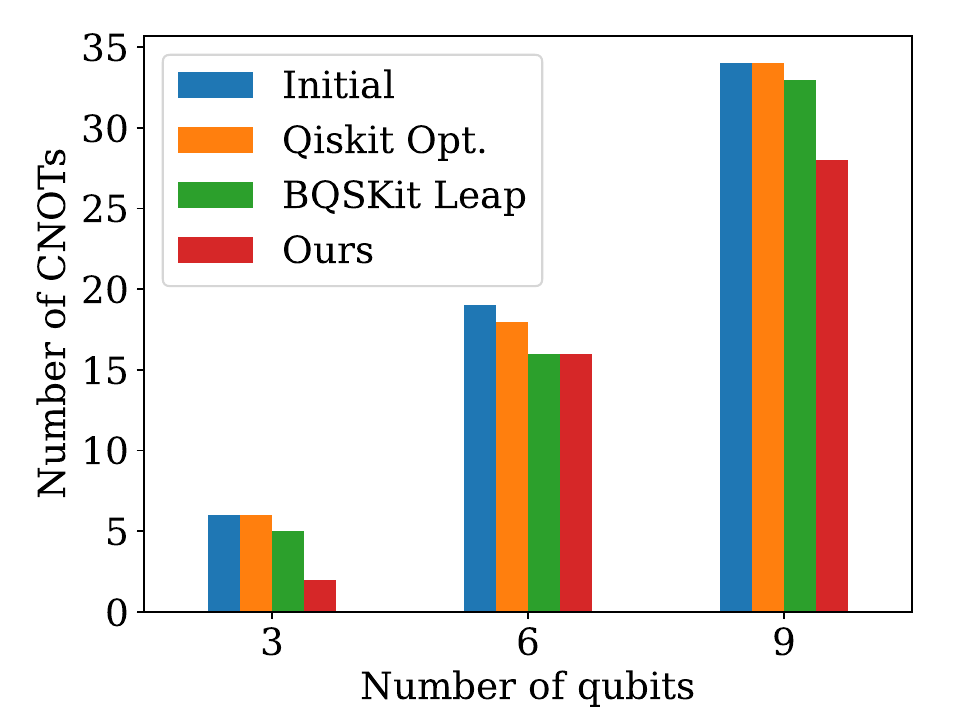}
        \caption{Random sparse uniform}\label{fig:rand-sparse-uniform}
    \end{subfigure}
    \begin{subfigure}[b]{.24\linewidth}
    \centering
        \includegraphics[width=.99\linewidth]{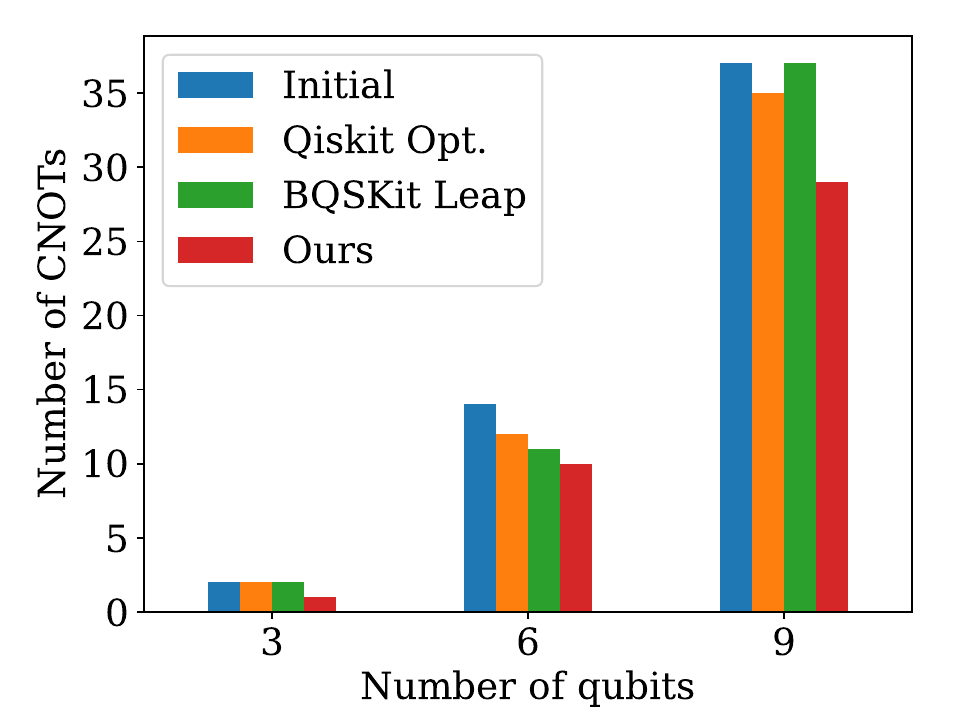}
        \caption{Random sparse}\label{fig:rand-sparse-nonuniform}
    \end{subfigure}
    \begin{subfigure}[b]{.24\linewidth}
    \centering
        \includegraphics[width=.99\linewidth]{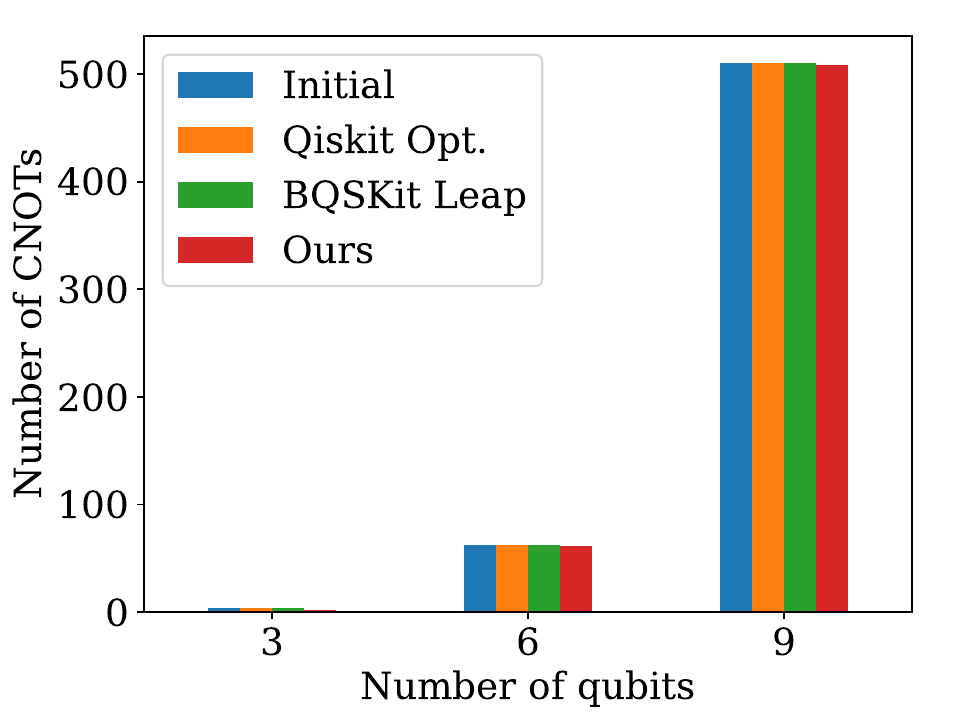}
        \caption{Random dense uniform}\label{fig:rand-dense-uniform}
    \end{subfigure}
    \begin{subfigure}[b]{.24\linewidth}
    \centering
        \includegraphics[width=.99\linewidth]{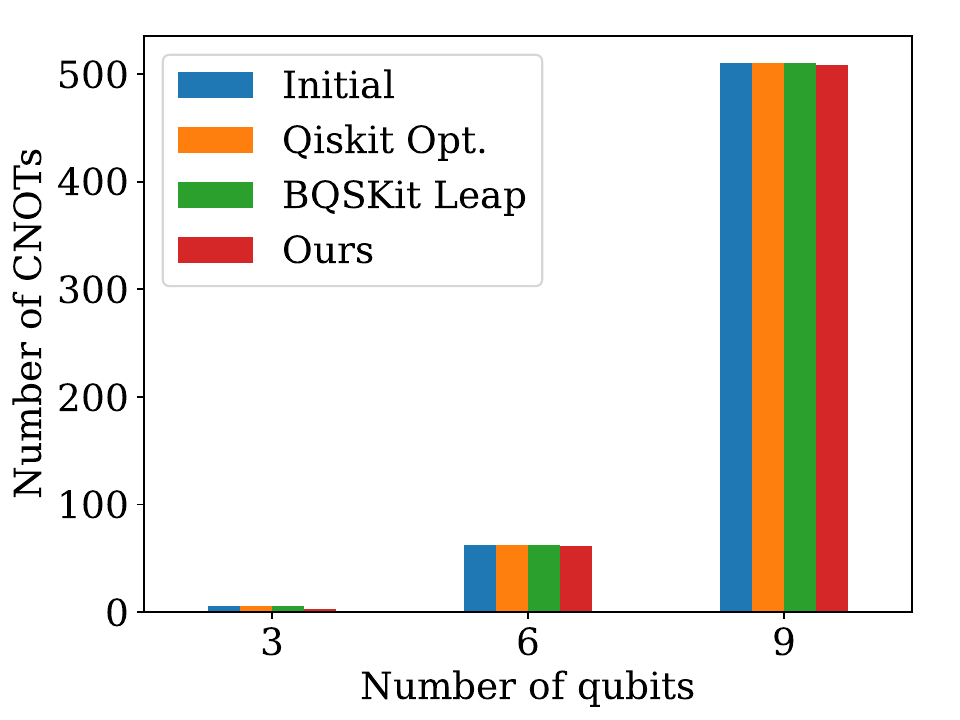}
        \caption{Random dense}\label{fig:rand-dense-nonuniform}
    \end{subfigure} \\
\vspace{.3cm}
    \begin{subfigure}[b]{.24\linewidth}
    \centering
        \includegraphics[width=.99\linewidth]{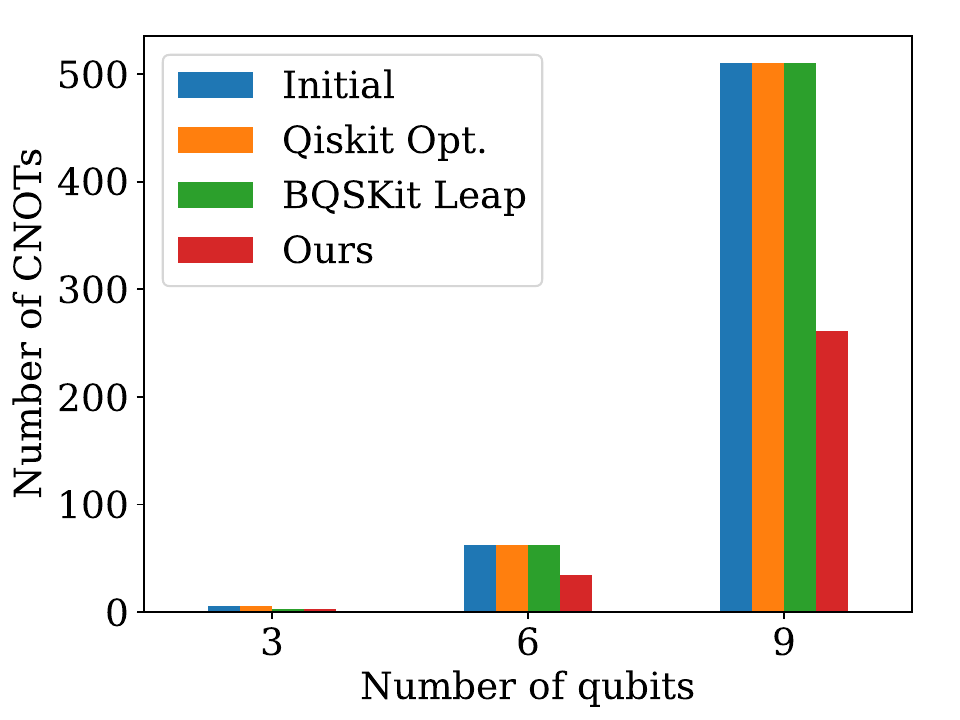}
        \caption{Dicke states $\ket{D_n^{\left \lceil{n/2}\right \rceil}}$}\label{fig:dicke}
    \end{subfigure}
    \begin{subfigure}[b]{.24\linewidth}
    \centering
        \includegraphics[width=.99\linewidth]{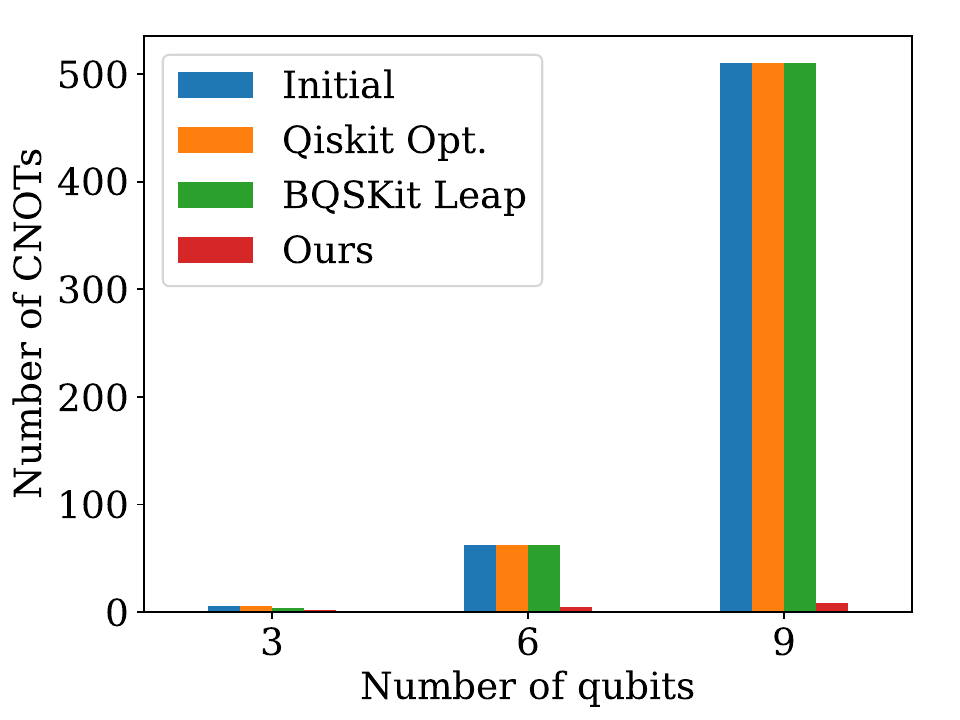}
        \caption{B states $\ket{B_n^{2^{(n-1)}+1}}$}\label{fig:qba}
    \end{subfigure}
    \begin{subfigure}[b]{.24\linewidth}
    \centering
        \includegraphics[width=.99\linewidth]{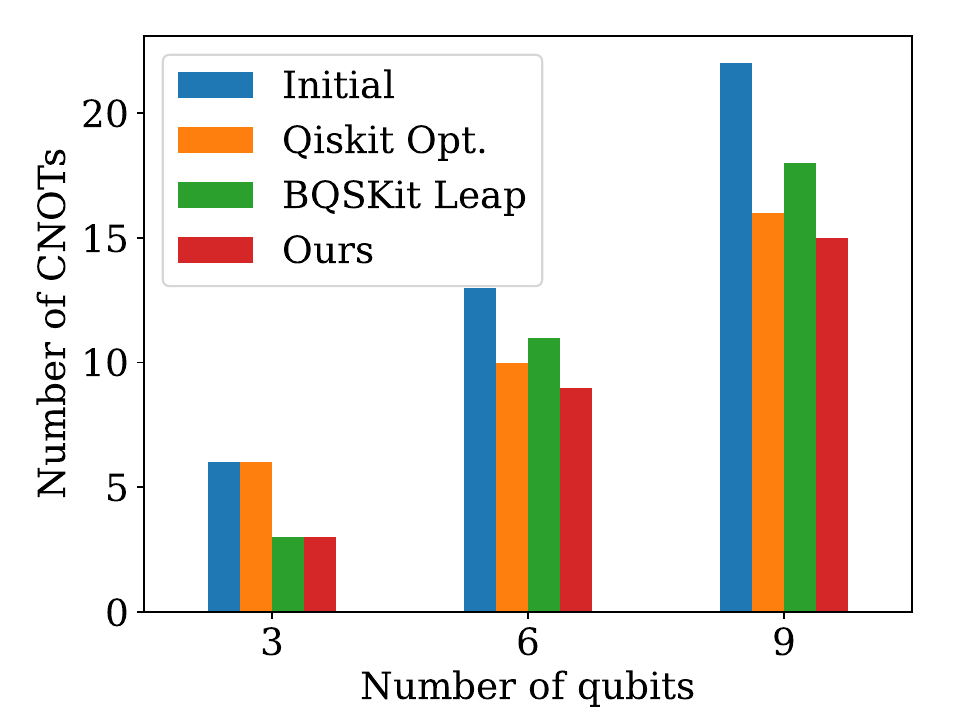}
        \caption{W states $\ket{W_n}$}\label{fig:w}
    \end{subfigure}
    \begin{subfigure}[b]{.24\linewidth}
    \centering
        \includegraphics[width=.99\linewidth]{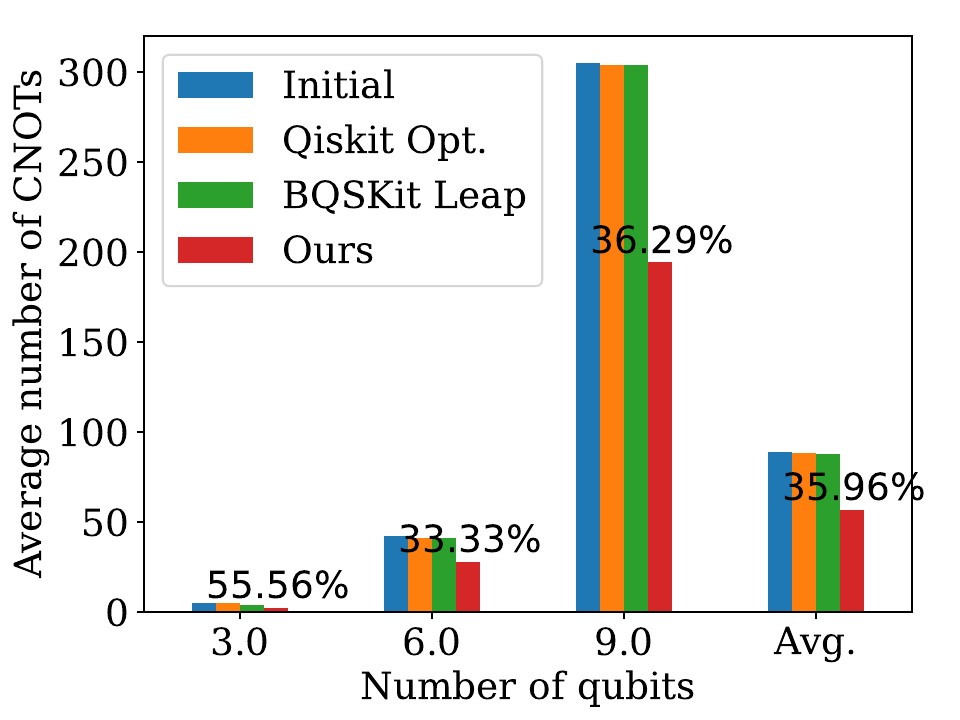}
        \caption{Average}\label{fig:avg}
    \end{subfigure}
    \caption{CNOT count comparison between three optimization algorithms.}\label{fig:cnot-comparison}
\end{figure*} 
\subsection{Complexity Analysis}
The overall resynthesis method is based on a peephole algorithm and is therefore scalable and efficient, running in polynomial time.

\begin{enumerate}
    \item Window extraction in Algorithm~\ref{alg:dc-resyn}. This step requires $\mathcal{O}(N^2)$ quantum operator simulations, where $N$ denotes the number of gates in the initial circuit.
    \item Resynthesis algorithm in Algorithm~\ref{alg:window-resyn}. For each segment, the time complexity is \tocheck{$\mathcal{O}(n^{K_{\max}}m^2K_{\max})$}, where $n$ is the number of qubits and $m$ is the cardinality of the care set in the extracted segment. This is because we need to enumerate $n^{K_{\max}}$ CNOT templates, and each template has a linear system with $K_{\max}+1$ columns and $m$ rows. We fix $K_{\max}=2$ to upper-bound the complexity to $\mathcal{O}(n^2m^2)$. 
\end{enumerate}

It is worth mentioning that we accelerate the quantum operator simulation to speed up our algorithm further. Although each operator corresponds to a unitary matrix with a dimension of $2^n$, the complexity, in practice, is lower because the implementation leverages sparsity and skips the multiplications of indices with zero amplitudes. Furthermore, all operators in the segment target the same qubits and do not affect the values of control qubits in the basis vectors. 
\section{Experimental Results}\label{sec:evaluation}
In this section, we present results to evaluate the effectiveness and efficiency of the proposed algorithm.

\subsection{Benchmarks and Workflow}
The benchmark suite used in our experiments consists of various states with different numbers of qubits, different levels of symmetry, and different sparsity. For each category of quantum states, we apply the most advantageous algorithm to generate the initial QSP circuit~\cite{mozafari2019preparation, gleinig2021efficient, wang2024quantum}. 
\begin{enumerate}
    \item $\ket{B_n^k}$ are quantum states with uniform amplitudes on the first $k+1$ basis states. $\ket{B_n^k} = \frac{1}{\sqrt{k+1}}\sum_{x=0}^k\ket{x}.$
    \item Dicke state $\ket{D_n^k}$ are highly symmetric $n$-qubit states. Indices of Dicke states have exactly $k$ ones and $n-k$ zeros~\cite{dicke1954coherence}. 
    \item W state $\ket{W_n}$ is a special case of Dicke states when $k=1$~\cite{dur2000three}.
    \item Sparse random states have cardinality $m=n$, and dense random states have $m=2^{n-1}$. 
\end{enumerate}

We implement the proposed method using Python and apply one iteration of the resynthesis algorithm introduced in \Cref{sec:methodology} to these initial QSP circuits. We compare the gate count reduction of our method with two circuit optimization methods as baselines, including the unitary synthesis pass in Qiskit (with an optimation level of 3)~\cite{Qiskit} and a powerful circuit optimization tool based on numerical methods~\cite{smith2023leap}. 

We map the optimized circuit from ours and baseline methods to $\libsmall$ and verify the correctness of the QSP circuits by running \texttt{qiskit} simulation~\cite{Qiskit}. All experiments are conducted on a computer equipped with a 3.7GHz AMD Ryzen 9 5900X processor and 64GB RAM.

Although our algorithm only applies to real unitary synthesis while baseline methods handle more general complex unitaries, we could adapt through phase oracle to prepare states with complex amplitudes. Besides, many interesting states, such as the benchmarks above, are real. Therefore, their state preparation circuits can be improved directly using our method. 

\subsection{CNOT Count Comparison}
\Cref{fig:cnot-comparison} displays the CNOT count comparison on seven categories of quantum states after running the three optimization algorithms. The average CNOT count shown in \Cref{fig:avg} demonstrates the advantages of our method in quantum circuit optimization. Leveraging the don't cares in the QSP circuit, we further lower the CNOT count of the initial circuit by 36\%. Meanwhile, the two baseline algorithms do not significantly optimize the circuit because of the constraints to completely preserve the circuit unitary. 
\begin{figure*}[t]
\centering
\mbox{
\small
    \Qcircuit @C=.4em @R=.3em {
\lstick{q_0} & \qw & \qw & \qw & \qw & \qw & \qw & \qw & \qw & \qw & \qw & \qw & \qw & \gate{\parbox{0.6cm}{\centering \footnotesize $R_y$\\$1.96$}} & \targ & \gate{\parbox{0.6cm}{\centering \footnotesize $R_y$\\$-\pi/8$}} & \targ & \gate{\parbox{0.6cm}{\centering \footnotesize $R_y$\\$-\pi/8$}} & \targ & \gate{\parbox{0.6cm}{\centering \footnotesize $R_y$\\$-\pi/8$}} & \targ & \gate{\parbox{0.6cm}{\centering \footnotesize $R_y$\\$\pi/8$}} & \targ & \gate{\parbox{0.6cm}{\centering \footnotesize $R_y$\\$\pi/8$}} & \targ & \gate{\parbox{0.6cm}{\centering \footnotesize $R_y$\\$\pi/8$}} & \targ & \gate{\parbox{0.6cm}{\centering \footnotesize $R_y$\\$-\pi/8$}} & \targ & \qw \\
\lstick{q_1} & \qw & \qw & \qw & \qw & \qw & \gate{\parbox{0.6cm}{\centering \footnotesize $R_y$\\$\pi/2$}} & \targ & \gate{\parbox{0.6cm}{\centering \footnotesize $R_y$\\$-\pi/4$}} & \targ & \gate{\parbox{0.6cm}{\centering \footnotesize $R_y$\\$\pi/4$}} & \targ & \targ & \qw & \ctrl{-1} & \qw & \qw & \qw & \ctrl{-1} & \qw & \qw & \qw & \ctrl{-1} & \qw & \qw & \qw & \ctrl{-1} & \qw & \qw & \qw \\
\lstick{q_2} & \qw & \gate{\parbox{0.6cm}{\centering \footnotesize $R_y$\\$\pi/4$}} & \targ & \gate{\parbox{0.6cm}{\centering \footnotesize $R_y$\\$\pi/4$}} & \targ & \qw & \ctrl{-1} & \qw & \qw & \qw & \ctrl{-1} & \qw & \qw & \qw & \qw & \ctrl{-2} & \qw & \qw & \qw & \qw & \qw & \qw & \qw & \ctrl{-2} & \qw & \qw & \qw & \qw & \qw \\
\lstick{q_3} & \gate{\parbox{0.6cm}{\centering \footnotesize $R_y$\\$0.68$}} & \qw & \ctrl{-1} & \qw & \ctrl{-1} & \qw & \qw & \qw & \ctrl{-2} & \qw & \qw & \ctrl{-2} & \qw & \qw & \qw & \qw & \qw & \qw & \qw & \ctrl{-3} & \qw & \qw & \qw & \qw & \qw & \qw & \qw & \ctrl{-3} & \qw
    }
}
\caption{$4$-qubit B state preparation using $14$ single-qubit gates and $14$ CNOTs.}\label{fig:qba-initial}
\vspace{-3mm}
\end{figure*}
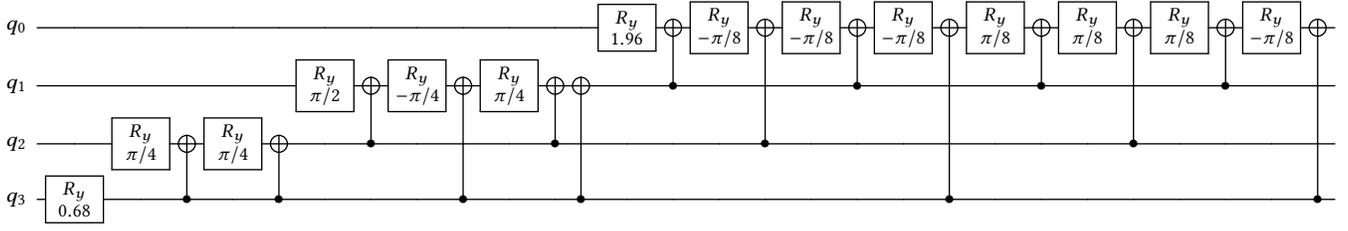
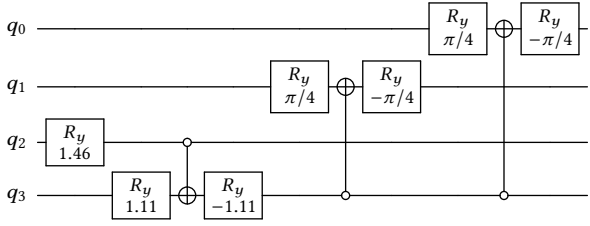
\begin{figure}[h]
\centering
\mbox{
\small
    \Qcircuit @C=.4em @R=.3em {
\lstick{q_0} & \qw & \qw & \qw & \qw & \qw & \qw & \qw & \gate{\parbox{0.6cm}{\centering \footnotesize $R_y$\\$\pi/4$}} & \targ & \gate{\parbox{0.6cm}{\centering \footnotesize $R_y$\\$-\pi/4$}} & \qw \\
\lstick{q_1} & \qw & \qw & \qw & \qw & \gate{\parbox{0.6cm}{\centering \footnotesize $R_y$\\$\pi/4$}} & \targ & \gate{\parbox{0.6cm}{\centering \footnotesize $R_y$\\$-\pi/4$}} & \qw & \qw & \qw & \qw \\
\lstick{q_2} & \gate{\parbox{0.6cm}{\centering \footnotesize $R_y$\\$1.46$}} & \qw & \ctrlo{1} & \qw & \qw & \qw & \qw & \qw & \qw & \qw & \qw \\
\lstick{q_3} & \qw & \gate{\parbox{0.6cm}{\centering \footnotesize $R_y$\\$1.11$}} & \targ & \gate{\parbox{0.6cm}{\centering \footnotesize $R_y$\\$-1.11$}} & \qw & \ctrlo{-2} & \qw & \qw & \ctrlo{-3} & \qw & \qw
    }
}
\caption{$4$-qubit B state preparation utilizing $7$ single-qubit gates and $3$ CNOTs.}\label{fig:qba-opt}
\end{figure}
\begin{figure}[b]
    \centering
    \begin{subfigure}[b]{.48\linewidth}
    \centering
        \includegraphics[width=.99\linewidth]{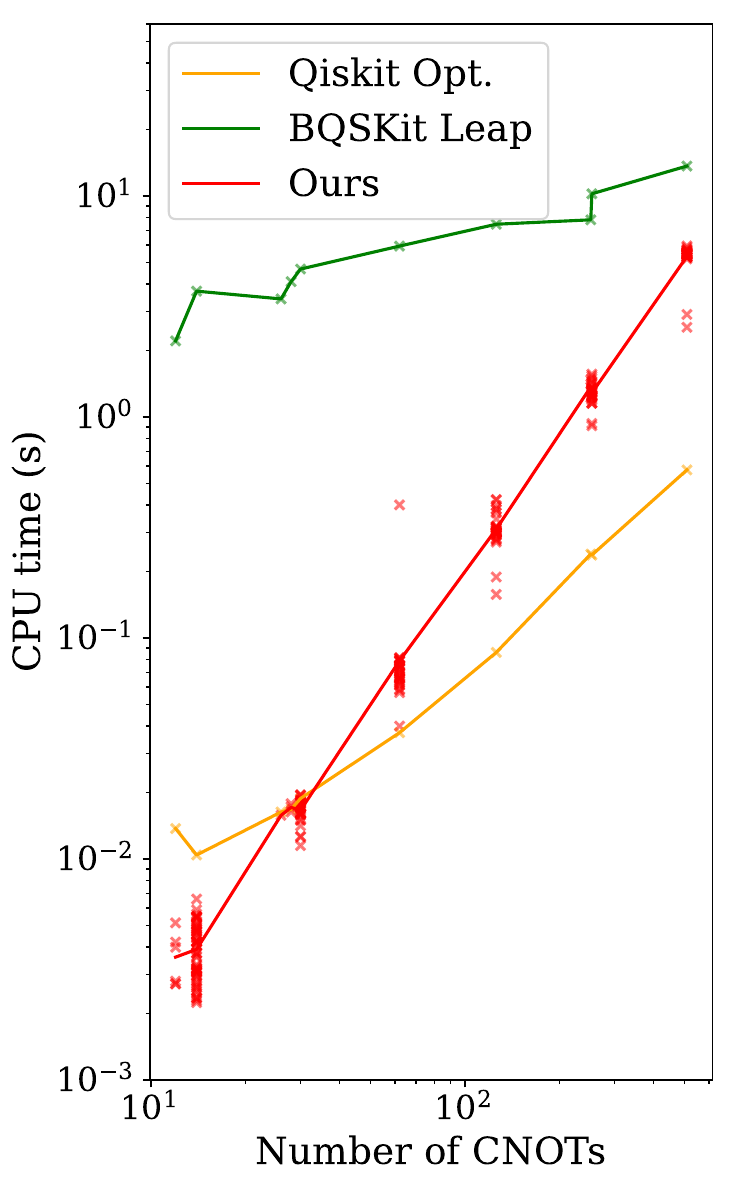}
        \caption{Dense states}\label{fig:resynthesis_time_dense}
    \end{subfigure}
    \begin{subfigure}[b]{.48\linewidth}
    \centering
        \includegraphics[width=.99\linewidth]{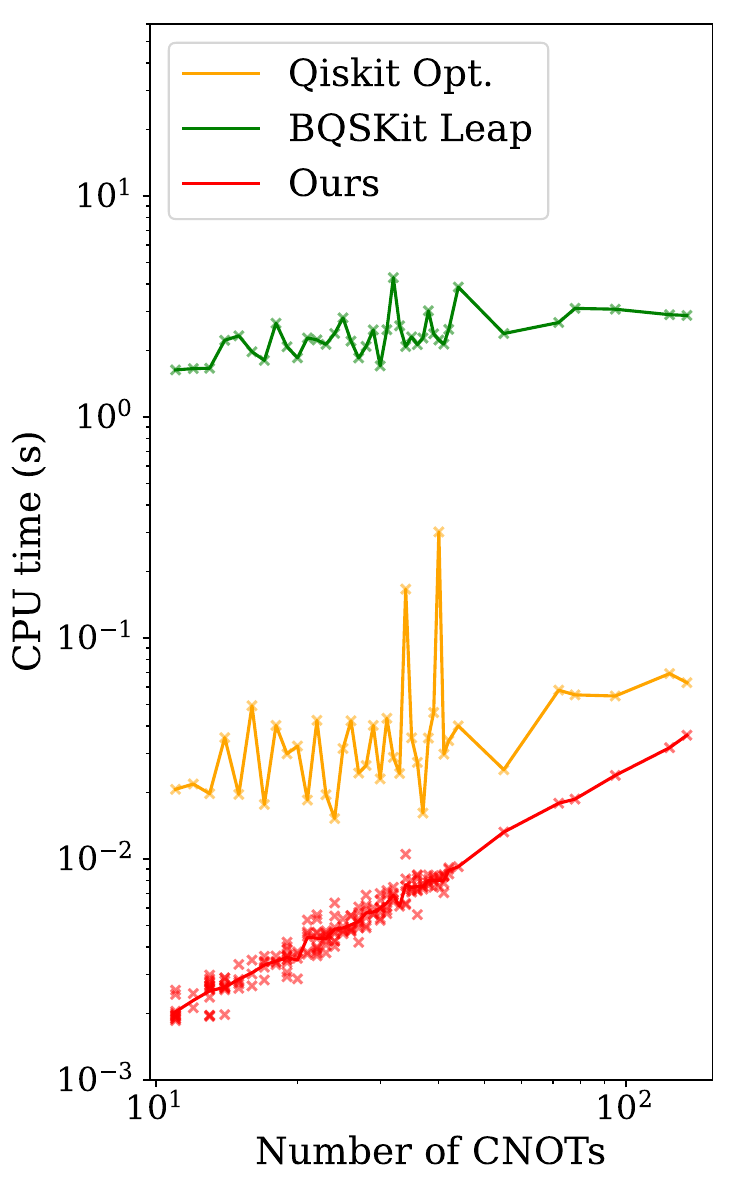}
        \caption{Sparse states}\label{fig:resynthesis_time_sparse}
    \end{subfigure}
    \caption{Scalability analysis. We demonstrate the relationship between CPU time on the $y$-axis and the number of CNOTs in the initial circuit on the $x$-axis.}\label{fig:results-scalability}
\end{figure}

In general, our algorithm is more effective on sparse states than on dense ones. This is because the number of don't cares decreases as the cardinality of the index set $m$ increases in our formulation. Note that our method simplifies the circuit by identifying an alternative unitary allowed by the don't-care conditions. Consequently, the greater the circuit's flexibility, the higher the likelihood of finding a feasible replacement that minimizes the gate count. 


A particularly interesting result is the optimization of $\ket{B_n^{2^{(n-1)}+1}}$ state preparation circuits. \Cref{fig:qba} exhibits that our resynthesis algorithm exponentially reduces the gate count required to prepare these states. To illustrate the optimization, we display the initial QSP circuit for $\ket{B_4^{9}}$ in \Cref{fig:qba-initial} and our optimized circuit in \Cref{fig:qba-opt}.

Notice that although both circuits subsequently entangle qubits into the system, the initial circuit in \Cref{fig:qba-initial} needs 2, 4, and 8 CNOTs to entangle $q_2$, $q_3$ and $q_4$, respectively, while \Cref{fig:qba-opt} uses 1 CNOT for each qubit. Since the segments to entangle these qubits are all single-target, the result confirmed that our algorithm successfully extracts the rotation tables of these segments and identifies the most efficient CNOT configuration to accomplish these rotations. 

\subsection{Scalablity Analysis}
The principal contribution of this paper is the introduction of a scalable resynthesis algorithm. This algorithm partitions the quantum circuit into single-target segments and uses the propagation of don't care conditions to derive a rotation table. By converting the complex problem of unitary synthesis, which typically involves matrix multiplications with high-order terms, into simpler linear equalities, our approach not only simplifies the process but also utilizes don't care conditions more efficiently, enhancing practical utility.

\Cref{fig:results-scalability} depicts the relationship between the CPU time required by three optimization algorithms and the number of CNOTs in the initial circuit. The runtime exhibits linear growth on a log-log plot, suggesting that our algorithm operates with polynomial complexity when limiting the value of $K_{\max}$, which is $2$ in this paper. Furthermore, our resynthesis algorithm completes circuits encompassing up to 10 qubits and several hundred CNOT gates within a few seconds, thereby proving both the efficiency and practical utility of our method.

\section{Conclusion}
In the noisy intermediate-scale quantum computing era, minimizing the count of two-qubit gates within quantum circuits is crucial. Due to the inherent complexity, existing algorithms typically assume a fixed unitary to optimize the circuit. However, modifying the unitary operator does not necessarily affect the functional correctness of the circuit in several quantum applications, and leveraging the flexibility can lead to significant circuit simplification. This paper introduces a scalable workflow that employs don't-care conditions to encode such flexibilities within the unitary matrix and develops an efficient resynthesis algorithm to optimize gate counts. Our experimental results demonstrate a significant reduction in CNOT gate usage. By permitting adjustments to the unitary during the optimization process, our method achieves an average decrease of 36\% in CNOT count when applied to the optimization of state preparation circuits.

\section*{Acknowledgements}
This work is funded by NSF grant 2313083. The authors would like to thank Marci Baun for proofreading the paper.

\newpage
\bibliographystyle{unsrt}
\bibliography{main.bib}

\begin{thebibliography}{10}

\bibitem{greenberger1989going}
Daniel~M Greenberger et~al.
\newblock Going beyond {Bell}'s theorem.
\newblock In {\em {Bell}’s Theorem, Quantum Theory and Conceptions of the
  Universe}, pages 69--72. Springer, 1989.

\bibitem{dur2000three}
Wolfgang D{\"u}r et~al.
\newblock Three qubits can be entangled in two inequivalent ways.
\newblock {\em Physical Review A}, 62(6):062314, 2000.

\bibitem{dicke1954coherence}
Robert~H {Dicke}.
\newblock Coherence in spontaneous radiation processes.
\newblock {\em Physical Review}, 93(1):99, 1954.

\bibitem{toth2012multipartite}
G{\'e}za T{\'o}th.
\newblock Multipartite entanglement and high-precision metrology.
\newblock {\em Physical Review A}, 85(2):022322, 2012.

\bibitem{murta2023preparing}
Bruno Murta, Pedro~MQ Cruz, and Joaqu{\'\i}n Fern{\'a}ndez-Rossier.
\newblock Preparing valence-bond-solid states on noisy intermediate-scale
  quantum computers.
\newblock {\em Physical Review Research}, 5(1):013190, 2023.

\bibitem{read1989valence}
N~Read and Subir Sachdev.
\newblock Valence-bond and spin-{P}eierls ground states of low-dimensional
  quantum antiferromagnets.
\newblock {\em Physical review letters}, 62(14):1694, 1989.

\bibitem{ashhab2022quantum}
Sahel Ashhab.
\newblock Quantum state preparation protocol for encoding classical data into
  the amplitudes of a quantum information processing register's wave function.
\newblock {\em Physical Review Research}, 4(1):013091, 2022.

\bibitem{araujo2021divide}
Israel~F Araujo et~al.
\newblock A divide-and-conquer algorithm for quantum state preparation.
\newblock {\em Scientific Reports}, 11(1):6329, 2021.

\bibitem{mozafari2019preparation}
Fereshte Mozafari, Mathias Soeken, and Giovanni De~Micheli.
\newblock Preparation of uniform quantum states utilizing {B}oolean functions.
\newblock In {\em 28th International Workshop on Logic Synthesis (IWLS)}, 2019.

\bibitem{niemann2016logic}
Philipp Niemann, Rhitam Datta, and Robert Wille.
\newblock Logic synthesis for quantum state generation.
\newblock In {\em 2016 IEEE 46th International Symposium on Multiple-Valued
  Logic (ISMVL)}, pages 247--252. IEEE, 2016.

\bibitem{gleinig2021efficient}
Niels Gleinig and Torsten Hoefler.
\newblock An efficient algorithm for sparse quantum state preparation.
\newblock In {\em 2021 58th ACM/IEEE Design Automation Conference (DAC)}, pages
  433--438. IEEE, 2021.

\bibitem{mozafari2022efficient}
Fereshte Mozafari et~al.
\newblock Efficient deterministic preparation of quantum states using decision
  diagrams.
\newblock {\em Physical Review A}, 106(2):022617, 2022.

\bibitem{malvetti2021quantum}
Emanuel Malvetti et~al.
\newblock Quantum circuits for sparse isometries.
\newblock {\em Quantum}, 5:412, 2021.

\bibitem{wu2020qgo}
Xin-Chuan Wu, Marc~Grau Davis, Frederic~T Chong, and Costin Iancu.
\newblock {QGo}: Scalable quantum circuit optimization using automated
  synthesis.
\newblock {\em arXiv preprint arXiv:2012.09835}, 2020.

\bibitem{davis2019heuristics}
Marc~Grau Davis, Ethan Smith, Ana Tudor, Koushik Sen, Irfan Siddiqi, and Costin
  Iancu.
\newblock Heuristics for quantum compiling with a continuous gate set.
\newblock {\em arXiv preprint arXiv:1912.02727}, 2019.

\bibitem{smith2023leap}
Ethan Smith, Marc~Grau Davis, Jeffrey Larson, Ed~Younis, Lindsay~Bassman
  Oftelie, Wim Lavrijsen, and Costin Iancu.
\newblock {LEAP}: Scaling numerical optimization based synthesis using an
  incremental approach.
\newblock {\em ACM Transactions on Quantum Computing}, 4(1):1--23, 2023.

\bibitem{davis2020towards}
Marc~G Davis, Ethan Smith, Ana Tudor, Koushik Sen, Irfan Siddiqi, and Costin
  Iancu.
\newblock Towards optimal topology aware quantum circuit synthesis.
\newblock In {\em 2020 IEEE International Conference on Quantum Computing and
  Engineering (QCE)}, pages 223--234. IEEE, 2020.

\bibitem{younis2020qfast}
Ed~Younis, Koushik Sen, Katherine Yelick, and Costin Iancu.
\newblock Qfast: Conflating search and numerical optimization for scalable
  quantum circuit synthesis.
\newblock In {\em 2021 IEEE International Conference on Quantum Computing and
  Engineering (QCE)}, pages 232--243, 2021.

\bibitem{kissinger2019pyzx}
Aleks Kissinger and John van~de Wetering.
\newblock {PyZX}: Large scale automated diagrammatic reasoning.
\newblock {\em arXiv preprint arXiv:1904.04735}, 2019.

\bibitem{staudacher2021optimization}
Korbinian Staudacher.
\newblock {\em Optimization Approaches for Quantum Circuits using ZX-calculus}.
\newblock PhD thesis, Master’s thesis, Ludwig-Maximilians-Universit{\"a}t,
  M{\"u}nchen., 2021.

\bibitem{szasz2023numerical}
Aaron Szasz, Ed~Younis, and Wibe De~Jong.
\newblock Numerical circuit synthesis and compilation for multi-state
  preparation.
\newblock In {\em 2023 IEEE International Conference on Quantum Computing and
  Engineering (QCE)}, volume~1, pages 768--778. IEEE, 2023.

\bibitem{ashhab2023quantum}
Sahel Ashhab, Fumiki Yoshihara, Miwako Tsuji, Mitsuhisa Sato, and Kouichi
  Semba.
\newblock Quantum circuit synthesis via a random combinatorial search.
\newblock {\em Phys. Rev. A}, 109:052605, May 2024.

\bibitem{nielsen2010quantum}
Michael~A Nielsen and Isaac~L Chuang.
\newblock {\em Quantum computation and quantum information}.
\newblock Cambridge University Press, 2010.

\bibitem{barenco1995elementary}
Adriano Barenco et~al.
\newblock Elementary gates for quantum computation.
\newblock {\em Physical Review A}, 52(5):3457, 1995.

\bibitem{krol2022efficient}
Anna~M Krol, Aritra Sarkar, Imran Ashraf, Zaid Al-Ars, and Koen Bertels.
\newblock Efficient decomposition of unitary matrices in quantum circuit
  compilers.
\newblock {\em Applied Sciences}, 12(2):759, 2022.

\bibitem{iccad20-tan-cong-optimal-layout-synthesis}
Bochen Tan and Jason Cong.
\newblock Optimal layout synthesis for quantum computing.
\newblock In {\em Proceedings of the 39th {IEEE}/{ACM} International Conference
  on Computer-Aided Design}, {ICCAD} '20, {Virtual Event, USA}, July 2020.
  Association for Computing Machinery.

\bibitem{Amy_2018}
Matthew Amy et~al.
\newblock On the controlled-not complexity of controlled-not-phase circuits.
\newblock {\em Quantum Science and Technology}, 4(1):015002, September.

\bibitem{Qiskit}
{Qiskit contributors}.
\newblock Qiskit: An open-source framework for quantum computing, 2023.

\bibitem{savoj1991use}
Hamid Savoj and Robert~K Brayton.
\newblock The use of observability and external don't cares for the
  simplification of multi-level networks.
\newblock In {\em Proceedings of the 27th ACM/IEEE Design Automation
  Conference}, pages 297--301, 1991.

\bibitem{bartlett1988multi}
Karen~A Bartlett, Robert~K Brayton, Gary~D Hachtel, Reily~M Jacoby,
  Christopher~R Morrison, Richard~L Rudell, Alberto Sangiovanni-Vincentelli,
  and A~Wang.
\newblock Multi-level logic minimization using implicit don't cares.
\newblock {\em IEEE Transactions on Computer-Aided Design of Integrated
  Circuits and Systems}, 7(6):723--740, 1988.

\bibitem{damiani1990observability}
Maurizio Damiani and Giovanni De~Micheli.
\newblock Observability don't care sets and boolean relations.
\newblock In {\em ICCAD}, volume~90, pages 502--505, 1990.

\bibitem{wang2024quantum}
Hanyu Wang, Jason Cong, and Giovanni De~Micheli.
\newblock Quantum state preparation using an exact cnot synthesis formulation.
\newblock In {\em Design, Automation and Test in Europe Conference (DATE)},
  pages 1--6, 2024.

\end{thebibliography}

\end{document}